\documentclass{lmcs}
\pdfoutput=1

\usepackage{lastpage}
\lmcsdoi{22}{4}{18}
\lmcsheading{}{\pageref{LastPage}}{}{}%
{Jun.~18,~2025}{Nov.~05,~2026}{}

\usepackage[utf8]{inputenc}

\mathchardef\hyphen=45 %

\newcommand{\re}[1]{\xrightarrow{#1}}
\newcommand{\rp}[1]{\overset{#1}{\rightsquigarrow}}

\newcommand{\tand}{\text{ and }}
\newcommand{\tor}{\text{ or }}
\newcommand{\tin}{\text{ in }}
\newcommand{\tif}{\text{ if }}
\newcommand{\tow}{\text{ otherwise}}

\newcommand{\leb}{\overline{<}}

\newcommand{\caprod}{\overset{\leftarrow}{\otimes}}
\newcommand{\caprodPow}[1]{\overset{\leftarrow {\scriptstyle #1}}{\otimes}}

\newcommand{\ap}{\overset{\leftarrow}{+}}

\newcommand{\ind}{\mathrm{ind}}
\newcommand{\out}{\mathit{out}}

\newcommand{\emptyword}{\epsilon}

\newcommand{\ASum}{\overleftarrow{\sum}}

\newcommand{\TL}{\mathrm{TL}}
\newcommand{\TW}{\mathrm{TW}}
\newcommand{\Parity}{\mathrm{Parity}}
\newcommand{\MaxParity}{\mathrm{MaxParity}}
\newcommand{\MinParity}{\mathrm{MinParity}}

\newcommand{\minLexProduct}[1]{\prod^{\mathrm{min\hyphen lex}}_{#1}}
\newcommand{\maxLexProduct}[1]{\prod^{\mathrm{max\hyphen lex}}_{#1}}

\newcommand{\eps}{\varepsilon}

\newcommand{\loopC}[1]{\tikz[baseline=(point.base)]{
	\node[] (point) {$\bullet$};
	\draw[->] (-0.07,0.075) to [out=120,in=60,looseness=4] node [anchor = south west,scale=0.7, inner sep=0.5mm] {#1}  (0.07,0.075);}}

\newcommand{\mininf}{\mathrm{mininf}}

\newcommand{\omegaBuchi}{\omega\text{-Büchi}}

\newcommand{\coBuchi}{\mathrm{coBuchi}}

\renewcommand{\ind}{\mathrm{ind}}

\newcommand{\boldclass}[3]{\ensuremath{\mathbf{#1}^{#2}_{#3}}}

\newcommand{\bsigma}[1]{\boldclass{\Sigma}{0}{#1}}
\newcommand{\bpi}[1]{\boldclass{\Pi}{0}{#1}}
\newcommand{\bdelta}[1]{\boldclass{\Delta}{0}{#1}}

\newcommand{\preimCB}[1]{\llbracket {#1} \rrbracket}
\newcommand{\gsmall}{g_\mathrm{small}}
\newcommand{\gbig}{g_\mathrm{big}}
\newcommand{\union}{\mathrm{union}}

\usepackage{logic9-thm}
\usepackage{todonotes}
\usepackage{stmaryrd}
\usepackage[shortcuts]{extdash}

\theoremstyle{plain}\newtheorem{claim}[thm]{Claim}
\theoremstyle{plain} 
    
\begin{document}

\title{Infinite lexicographic products of positional objectives} 

\author[A. Casares]{Antonio Casares\lmcsorcid{https://orcid.org/0000-0002-6539-2020}}[a]
\author[P. Ohlmann]{Pierre Ohlmann\lmcsorcid{https://orcid.org/0000-0002-4685-5253}}[b]
\author[M. Skrzypczak]{Michał Skrzypczak\lmcsorcid{https://orcid.org/0000-0002-9647-4993}}[c]
\author[I. Walukiewicz]{Igor Walukiewicz\lmcsorcid{https://orcid.org/0000-0001-8952-7201}}[d]

\address{University of Kaiserslautern-Landau}
\email{antonio.casares@lis-lab.fr}
\address{CNRS, LIS, Université Aix-Marseille}
\email{pierre.ohlmann@lis-lab.fr}
\address{University of Warsaw}
\email{mskrzypczak@mimuw.edu.pl}
\address{CNRS, LaBRI, Université de Bordeaux}
\email{igw@labri.fr}

\thanks{Antonio Casares is partially supported by the Deutsche Forschungsgemeinschaft (grant number 522843867) and the European Research Council (grant number 101089343).
Part of this work was done while Casares was at the University of Warsaw, Poland, supported by the Polish National Science Centre (NCN) grant ``Polynomial finite state computation'' (2022/46/A/ST6/00072).\\
Michał Skrzypczak was supported by the National Science Centre, Poland (grant no.\@ 2021/\allowbreak41/\allowbreak B/\allowbreak ST6/\allowbreak03914).}

\begin{abstract}
This paper contributes to the study of positional determinacy of infinite-duration games played on potentially infinite graphs with neutral transitions.
Recently, [Ohlmann, TheoretiCS 2023] established that positionality of prefix-independent objectives is preserved by finite lexicographic products.
We propose two different notions of infinite lexicographic products indexed by ordinals, and extend Ohlmann's result by proving that they also preserve positionality.
In the context of one-player positionality, this extends positional determinacy results of [Grädel and Walukiewicz, Logical Methods in Computer Science 2006] to edge-labelled games and to countably many priorities for Max-Parity and arbitrarily many priorities for Min-Parity.
Moreover, we show that the Max-Parity objectives over countable ordinals are complete for the infinite levels of the difference hierarchy over $\bsigma{2}$ and that Min-Parity is complete for the class $\bsigma{3}$.
We therefore obtain positional languages that are complete for all these levels, as well as new insights into closure under unions and neutral letters.
\end{abstract}

\maketitle

\paragraph{AI statement.} In the final phase of publication (after peer-reviewing), this manuscript was submitted to Claude (Anthropic) for feedback, which led to a few minor technical improvements as well as the removal of some flawed material.

\section{Introduction}\label{sec:intro}

\subsection{Context: Positionality in games on graphs}
We consider infinite-duration games played on directed
graphs whose edges are labelled by colours from a set $C$, with a specified
objective $W \subseteq C^\omega$.
Both the game graph and the set of colours may be infinite.
The two players, Eve and Adam, take turns moving a token along the edges of the graph.
If the sequence of colours appearing on the produced path belongs to $W$, then Eve wins; otherwise, Adam wins.
If the objective $W$ is Borel, then the game is determined, meaning that one of the two
players has a winning strategy~\cite{Martin75}.

This paper is part of a long line of research aiming to understand which Borel objectives are \emph{positional}.
A positional strategy depends only on the current vertex of the game and not on
the whole history of the play so far. 
An objective is positional for Eve (just positional\footnote{In some parts of the literature, these are called half-positional or memoryless for Eve.} in the following) if whenever Eve has a winning strategy in a
game with this objective, then she has a positional one. 

Recently, Ohlmann~\cite{Ohlmann23} introduced universal graphs to the study of positional
objectives, and proved that an objective\footnote{All objectives in this paper
are prefix-independent and admit a neutral letter, as explained in Section~\ref{sec:prelims}.} $W$ is positional  if and only if it admits monotone well-ordered universal graphs (see Section~\ref{sec:prelims} for formal definitions).
This result has been generalised to characterise the memory of objectives~\cite{CO25}, and universal graphs have already proved key to deciding the positionality of $\omega$-regular objectives~\cite{BCRV24HalfJournal,CO24Positional} and computing their memory~\cite{CO25memory}.
Universal graphs are the central object of study in this work.  

\paragraph{Closure properties and lexicographic products.}
Some of the questions surrounding positional objectives concern their closure properties, with two major open problems in the area focusing on this aspect:  
\begin{itemize}
    \item Kopczyński's Conjecture~\cite[Conjecture~7.1]{Kop08Thesis}: Are positional objectives closed under finite and countable unions? This question has been answered positively for countable unions of  $\bsigma{2}$ objectives~\cite[Corollary~3]{OS24Sigma2} and for finite unions of their Boolean combinations (including all $\omega$-regular objectives)~\cite[Theorem~12]{CO25memory} and negatively for positionality over finite graphs~\cite{Kozachinskiy24EnergyGroups}.
    \item Neutral Letter Conjecture~\cite{Ohlmann23}: Are positional objectives closed under the addition of a neutral letter, that is, a letter
    whose addition or removal from a word $w$ does not change whether $w$ belongs to
    $W$? This conjecture has important consequences for the completeness of the characterisation of positionality via universal graphs (see~\cite{Ohlmann23} or Section~\ref{sec:prelims} for details).
\end{itemize}

One of the few known closure properties of positional objectives is given by \emph{finite lexicographic products}, obtained as a corollary of the characterisation based on universal graphs~\cite{Ohlmann23}.
The lexicographic product of a sequence of objectives $(W_i\subseteq C_i^\omega)_{i<k}$ is their hierarchical combination: a word $w$ belongs to the product if $\pi_i(w)\in W_i$, where $i$ is the largest index such that $w$ contains infinitely many occurrences of colours from $C_i$, and $\pi_i(w)$ is the subword obtained by restricting $w$ to these colours.
This hierarchical combination of objectives naturally arises from quantifier alternation in some logics, such as the alternation of fixpoint operators in the modal $\mu$-calculus.

A paradigmatic example of such a hierarchical construction is given by the parity objective
\[
    \Parity_d=\{w \in \{0,1,\dots, d\}^\omega \mid \limsup w \text{ is even}\},
\]
which enjoys a special status: it is one of the first objectives shown to be positional over arbitrary game graphs~\cite{EJ91,Mostowski91Forbidden}, a result which is central in modern proofs of Rabin's Theorem on the decidability of the logic S2S~\cite{Rabin69,GTW2002}, as well as in the algorithmic study of infinite-duration games~\cite{NathBook}.
It holds that the parity objective can be obtained as a finite lexicographic product of trivial objectives, giving an alternative positionality proof and highlighting the fundamental role of lexicographic products in the theory of positionality.

\paragraph{From finite to infinite products.}
A natural goal is to extend the previous ideas to infinite sequences of objectives.
The simplest example of such a construction is the Min-Parity objective over $\omega$, defined by
\[
    \MinParity_\omega = \{w \in \omega^\omega \mid \liminf w \text{ is finite and even}\}.
\]
$\MinParity_\omega$ was first studied by Grädel and Walukiewicz~\cite{GW06}, who established its bi\=/positionality, that is, positionality for both the objective and its complement. This result was proved for vertex-labelled game graphs.
Here, the distinction between vertex labels and edge labels is crucial; in fact,
it is easy to see (see Figure~\ref{fig:counter_example_intro}) that
$\MinParity_\omega$ is not positional for the opponent when edge labels are
considered.\footnote{It is easy to encode vertex-labels into edge-labels, and
therefore if an objective is edge-labelled positional then it is vertex-labelled
positional; but the converse is not true.}
Grädel and Walukiewicz~\cite{GW06} also observed that bi-positionality does not hold for $\MaxParity_\omega$ or $\MinParity_{\alpha}$ when $\alpha>\omega$.
However, failure of bi-positionality in these cases is due to phenomena akin to Figure~\ref{fig:counter_example_intro}: playing an increasing sequence of priorities requires memory, and therefore Adam requires memory.
Positionality over edge-labelled graphs for all these objectives was neither proved nor disproved in their work.

\begin{figure}[h]
\begin{center}
\includegraphics[width=0.19\linewidth]{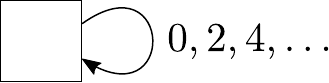}
\end{center}
\caption{An edge-labelled game controlled by Adam in which he requires non-positional strategies to ensure that $\MinParity_\omega$ is not met.}\label{fig:counter_example_intro}
\end{figure}

\paragraph{Topological complexity.} In this paper, we are also interested in the topological complexity of the new objectives that we obtain.
One important property of the parity objectives over finitely many colours is that they are Wadge-complete for the finite levels of the difference hierarchy over $\bsigma{2}$~\cite{mskrzypczak_colorings}.
The difference hierarchy consists of $\omega_1$ many levels of classes of sets, which all lie below $\bdelta{3}$. The main interest of this hierarchy is that it spans the whole class $\bdelta{3}$ (see for instance~\cite[Theorem~22.27]{Kechris1995}).
However, to the best of our knowledge, no natural positional languages were known for the infinite levels of this hierarchy (those between $\omega$ and $\omega_1$).

Another motivation for the study of infinite lexicographic products is the
development of tools for constructing new, complex positional objectives.
Since all $\omega$-regular objectives lie below $\bdelta{3}$, no positional objective was known above this class.\footnote{During the preparation of this manuscript, a positional $\bpi{3}$-complete objective has also been proposed~\cite{COV24Pi3}.\label{footnote:COV24Pi3}}
Indeed, an important obstacle to advance on Kopczyński's Conjecture and the Neutral Letter Conjecture is the lack of such tools. %
In fact, as we will see, some candidate objectives for disproving the Neutral Letter Conjecture can be described in this framework, namely $\MinParity_\omega$ and  the $\omega$-Büchi objective, defined by 
\[
    \omegaBuchi = \{w \in \omega^\omega \mid |w|_i \text{ is infinite for some } i\}.
\]
Both these objectives are $\bsigma{3}$-complete. 
Their positionality was not known prior to our work, and they are drastically altered when a neutral letter is added.

\subsection{Contributions}
We provide two ways of defining lexicographic products of families of objectives
indexed by ordinals, namely, the max\=/lexicographic product and the
min\=/lexicographic product.
We show that these operations preserve positionality (Theorems~\ref{thm:main_max_lex} and~\ref{th:main-pos-min-lex}).
The proofs rely on providing adequate constructions of well-ordered monotone
universal graphs for the two lexicographic products. 
We now discuss some further results and consequences. %

\paragraph{Topological completeness results.}
We study the objective\footnote{An ordinal is odd if it can be written as $\beta + n$, with $\beta$ either $0$ or a~limit ordinal and $n<\omega$ odd. The use of odd ordinals is crucial in this
definition, the reason being that limit ordinals are even and should be rejected for positionality---see Remark~\ref{rem:why-odd}.}
\[
    \MaxParity_\alpha = \{w \in \alpha^\omega \mid \limsup w \text{ is odd}\},
\]
for countable ordinals $\alpha$.
Extending the results of~\cite{mskrzypczak_colorings}, we prove completeness of $\MaxParity_\alpha$ for the corresponding level in the
difference hierarchy over~$\bsigma{2}$ (Theorem~\ref{thm:difference-complete}).
So, for infinitely many levels of the difference hierarchy spanning the whole $\bdelta{3}$, we obtain natural positional objectives complete for these levels. 
We believe that defining such a~class of complete objectives which are positional and admit a simple universal graph (see Section~\ref{sec:max-lex}) is key to achieving a complete understanding of positionality within $\bdelta{3}$, which is still elusive.

On the other hand, min-lexicographic products of trivial objectives can go beyond $\bdelta{3}$.
This is the case of $\omega$-Büchi and the $\MinParity_\alpha$ objectives, which are Wadge-complete for $\bsigma{3}$ for countably infinite $\alpha$ (Theorem~\ref{thm:sigma3-complete}). As far as we are aware, these are the first known positional objectives in this class.
This is a first step towards exploring positionality beyond~$\bdelta{3}$.

\paragraph{Closure under addition of neutral letters for some objectives.}
If an objective admits a well-ordered monotone universal graph, then not only is
it positional, but its extension with a neutral letter is positional as
well~\cite{Ohlmann23} (conversely, the restriction of a positional objective to a
subset of colours always remains positional). 
Therefore, all positionality results presented in this paper hold for both the objectives and their extensions with neutral letters. 
This is in particular the case for $\omega$-Büchi and $\MinParity_\alpha$ (for
any ordinal $\alpha$), but these two conditions were, until now, the best
potential candidates for disproving the Neutral Letter Conjecture.
Thus, our results suggest that adding neutral letters may preserve positionality in general.

\paragraph{Structure of the paper.}
We first recall the necessary definitions, including finite lexicographic products, in Section~\ref{sec:prelims}.
Then we present max-lexicographic products in Section~\ref{sec:max-lex} and (the more complex) min-lexicographic products in Section~\ref{sec:min-lex}.

\section{Preliminaries and finite lexicographic products}\label{sec:prelims}

\paragraph{Graphs.} In this paper, graphs are directed, edge-coloured, typically infinite, and may have sinks (vertices with no outgoing edges).
Formally, a \emph{$C$-graph} $G$, where $C$ is an arbitrary set of colours, is given by
a set of vertices $V(G)$ and a set of edges $E(G) \subseteq V(G) \times C \times
V(G)$. 
We will usually denote edges by $v \re c v'$.
A \emph{path} in a graph $G$ is a sequence of edges in $E(G)$ with matching endpoints, 
\[
    v_0 \re {c_0} v_1 \re{c_1} v_2 \re{c_2} \dots
\]
A path can be finite (even empty) or infinite.
We say that it is a path \emph{from} $v_0$, and, if it is finite and contains $i$ edges, \emph{to} $v_{i}$.
The finite or infinite word $c_0c_1 \dots$ is called the \emph{label} of the path.
When there is a path from $v$ to $v'$, we say that $v'$ is \emph{reachable} from $v$ in $G$.

A \emph{morphism} between two $C$-graphs $G$ and $H$ is a map $\phi:V(G) \to V(H)$ such that for every edge $v \re c v' \in E(G)$, it holds that $\phi(v) \re c \phi(v')$ is an edge in $E(H)$.
We write $G \to H$ if we just want to state the existence of such a morphism.
Note that $\phi$ need not be injective or surjective.
The composition of morphisms is a morphism.
A \emph{subgraph} $G'$ of $G$ is obtained from $G$ by removing vertices and edges of $G$; note that in that case $G' \to G$.
If $R \subseteq V(G)$, the subgraph of $G$ obtained by removing all vertices in $V(G) \setminus R$ and keeping all edges between vertices in $R$ is called the \emph{restriction of $G$ to $R$}.
Given a vertex $v \in V(G)$, we let $G[v]$ denote the restriction of $G$ to vertices reachable from $v$.
The size of a graph $G$ is the cardinal $|V(G)|$.
If $v \re c v' \in E(G)$ then we say that $v$ is a~\emph{$c$-predecessor} of $v'$ and $v'$ is a~\emph{$c$-successor} of $v$.
An edge $v \re c v$ is called a~\emph{loop} around $v$.

\paragraph{Ordered graphs, monotonicity, and directed sums.} We will often consider \emph{ordered graphs}, which are pairs $(G,\geq)$ where $\geq$ is a (partial) order on $V(G)$.
By a slight abuse of notation, we sometimes omit $\geq$ from the notation of an ordered graph.
We will pay special attention to graphs in which the order satisfies some of the following properties:
\begin{itemize}
	\item is total,
	\item is well-founded (any non-empty subset has a minimal
	element),
	\item is a well-order (total and well-founded).
\end{itemize}

\noindent
An ordered $C$-graph $(G,\geq)$ is said to be \emph{monotone} if for all $u,v,u',v' \in V(G)$ and $c \in C$ we have
\[
    u \geq v \re c v' \geq u' \tin G \quad \implies \quad u \re c u'.
\]
In proofs, it is sometimes convenient to break monotonicity into left-monotonicity ($u \geq v \re c v' \implies u \re c v'$) and right-monotonicity ($v \re c v' \geq u' \implies v \re c u'$); it is a direct check that monotonicity is equivalent to their conjunction.

Given a family of (ordered) $C$-graphs $(G_\mu)_{\mu < \alpha}$, where $\alpha$ is an
arbitrary ordinal, we define their \emph{directed sum} $\ASum_{\m<\alpha}G_\m$ to be 
the disjoint union of the $G_\mu$'s with added edges from each $G_\m$ to all the
graphs before it in the sequence.
Formally, $G=\ASum_{\m<\alpha}G_\m$  is a graph with vertices $V(G) = \bigsqcup_{\mu <
\alpha} V(G_\mu) \times \{\mu\}$, and edges
\begin{equation*}
  (v,\mu) \re c (v',\mu')\in E(G)\quad \tif \m>\m', \tor \; [\m=\m' \tand v \re c v' \in E(G_\mu)].
\end{equation*}
Note that for all $\mu$, it holds that $G_\mu \to G$.
If the $G_\mu$'s are ordered, then so is their sum, by the order
\[
    (v,\mu) \geq (v',\mu') \quad \tif \text{$\mu > \mu'$ or [$\mu=\mu'$ and $v \geq v' \tin G_\mu$]}.
\]
Observe that for any property $X$ among being totally ordered, well-founded, or monotone, if the $G_\mu$'s have property $X$ then so does their directed sum.
By a slight abuse, when the $V(G_\mu)$'s are disjoint sets, we define for convenience the sum over $\bigsqcup_{\mu < \alpha} V(G_\mu)$ instead of $\bigsqcup_{\mu < \alpha} V(G_\mu) \times \{\mu\}$.
In the case where the $G_{\mu}$'s are all equal to some (ordered) graph $G$, we denote their directed sum by $G \caprod \alpha$.

\paragraph{Objectives and universality.}
A \emph{$C$-objective} is a language\footnote{Formally, an objective is a pair $(C,W)$, where $C$ is non-empty and $W \subseteq C^\omega$. For simplicity, we just write objectives as $W$, as this does not create confusion or ambiguity.} of infinite words  $W\incl C^\w$.
In this paper, we will consider only prefix-independent objectives,
meaning those such that $cW = W$ for all $c\in C$ (equivalently, membership of a word in $W$ is not affected by addition or removal of a finite prefix).

We say that a $C$-graph \emph{$G$ satisfies an objective $W$} if the label of any of its infinite paths belongs to $W$.
In particular, a graph without infinite paths satisfies any objective. 

Given a cardinal $\kappa$, we say that a $C$-graph $U$ is \emph{$\kappa$-universal for $W$} if
\begin{itemize}
\item $U$ satisfies $W$; and
\item every graph $G$ of size $<\k$ and satisfying $W$ admits a morphism to $U$, i.e.~$G \to U$.
\end{itemize}
We also sometimes say that such a graph is $(\kappa,W)$-universal.

The next lemma indicates that, for prefix-independent objectives, it is in fact sufficient to find (monotone, well-ordered) universal graphs with weaker requirements.

\begin{lemC}[{\cite[Lemma~4.5]{Ohlmann23}}]\label{lemma:almost-universal}
	Let $W$ be a prefix-independent $C$-objective, $\kappa$ a cardinal, and $U$ be a $C$-graph such that:
	\begin{itemize}
		\item $U$ satisfies $W$; and
		\item for all graphs $G$ which satisfy $W$ and have size $<\kappa$, there is a vertex $v \in V(G)$ such that $G[v] \to U$.
	\end{itemize}
	Then $U \caprod \kappa$ is $\kappa$-universal for $W$.
\end{lemC}

\noindent
Following~\cite{Ohlmann23}, we sometimes say that a graph $U$ as above is almost $(\kappa,W)$-universal.

\begin{proof}
Let $U$ be such a graph and let $G$ be a graph of size $<\kappa$ satisfying $W$; we should prove that $G \to U \caprod \kappa$.
By hypothesis, there is a vertex $v_0$ such that $G[v_0] \to U$.

Now let $\lambda$ be any ordinal and assume that vertices $v_\mu$ have been constructed for $\mu<\lambda$.
Then we let $G_\lambda = G \setminus \bigcup_{\mu<\lambda} G[v_\mu]$ be the restriction of $G$ to vertices which are not reachable from any of the $v_\mu$'s.
Since $|G_\lambda| < \kappa$, there is $v_\lambda$ such that $G_\lambda[v_\lambda] \to U$.

Now note that $G$ is the disjoint union of the $G_\lambda[v_\lambda]$'s, and moreover $G_\lambda$ is empty if $\lambda \geq \kappa$.
Moreover, any edge in $G$ is either part of some $G_\lambda[v_\lambda]$, or goes from $G_\lambda[v_\lambda]$ to $G_{\lambda'}[v_{\lambda'}]$ for some $\lambda>\lambda'$.
We conclude that $G \to U \caprod \kappa$ by mapping $G_\lambda[v_\lambda]$ into the $\lambda$-th copy of $U$ for each $\lambda$.
\end{proof}

\paragraph{Universal graphs for the study of positionality.}
We introduce definitions of games and positionality for completeness. However, throughout the paper we will study positionality through the lens of universal graphs (by using Theorem~\ref{th:positional=UnivGraphs} below), and will not directly use the game-based definition of positionality.

A \emph{$W$-game} is given by a sinkless $C\cup\{\eps\}$-graph, together with a
(prefix-independent) $C$-objective $W$ and a partition of the 
vertices into those controlled by one player, called Eve, and those controlled by her adversary,
called Adam. 
Players play by moving a token in the graph indefinitely; the
player controlling the current vertex chooses which edge to take. 
The result of a play is an infinite path in the game graph.
Who wins the play is determined by the projection of the labels onto
$C$: Eve wins if this projection is finite or belongs to $W$; otherwise, Adam wins. 
This definition makes $\e$ a neutral letter.
A strategy (for Eve) is a function assigning to each finite path ending in a vertex controlled by Eve the next edge she should take.
Such a strategy is winning from a vertex $v$ if all infinite paths from $v$ following the strategy are winning. 

A strategy is \emph{positional} if it can be described by a function from the set of Eve's vertices to edges; the strategy always points to the same outgoing edge, independently of the past of the play.
An objective $W$ is \emph{positional}\footnote{It is not known whether the
presence of a neutral colour $\eps$ affects positionality.
Theorem~\ref{th:positional=UnivGraphs} concerns positionality in the presence of
a neutral colour due to the way we have defined games here.} if for every $W$-game, Eve has a positional strategy $\sigma$
such that if she has a winning strategy from a vertex $v$, she wins from $v$ using strategy
$\sigma$. 

\begin{thmC}[{\cite[Theorem~3.1]{Ohlmann23}}]\label{th:positional=UnivGraphs}
	A prefix-independent objective $W$ is positional if and only if for every cardinal $\kappa$ there exists a well-ordered monotone $\kappa$-universal graph for~$W$.
\end{thmC}

\noindent
In the following, we will use the term ``positional objective'' as a synonym for an objective admitting well-ordered monotone $\kappa$-universal graphs for all $\kappa$.

\paragraph{Trivial objectives.} For a non-empty set of colours $C$, we call
$\TW_C = C^\omega$ the trivially winning objective, and $\TL_C = \emptyset
\subseteq C^\omega$ the trivially losing objective over $C$.
We will write $\TW_c$ and $\TL_c$ if $C$ is the singleton $\{c\}$.
These objectives are positional: it is easy to see that the single-vertex
$C$-graph $\loopC C$ with all possible loops is $\kappa$-universal for $\TW_C$
for all $\kappa$.
For $\TL_C$, the graph of the order relation for cardinal
$\k$ is $\k$-universal.
This graph, which we denote by $\bullet \caprodPow{C}  \kappa$, has all ordinals $<\k$ as its vertices and contains an edge $\lambda \re c \lambda'$ for every $c\in C$ and ordinals $\lambda>\lambda'$.

\paragraph{Finite lexicographic products of objectives.} Let $C_0$ and $C_1$ be two disjoint sets of colours, and let $C=C_0 \cup C_1$.
Given an infinite word $w \in C^\omega$ and $i \in \{0,1\}$, we let $\pi_i(w)$ denote the (finite or infinite) word obtained by restricting $w$ to letters in $C_i$.

We then define the \emph{max\=/lexicographic product} of two prefix-independent
objectives $W_0 \subseteq C_0^\omega$ and $W_1  \subseteq C_1^\omega$ by 
\[
    \begin{aligned}
    W_0 \rtimes W_1 = \{w \in C^\omega \mid [\pi_1(w) \text{ is infinite and belongs to } W_1] \qquad \qquad \\ 
    \tor  [\pi_1(w) \text{ is finite and } \pi_0(w) \in W_0]\}.
    \end{aligned}
\]
Note that $W_0 \rtimes W_1$ is prefix-independent.
This operation is associative, and the max\=/lexicographic product of $p+1$
conditions is
\[
    W_0 \rtimes \dots \rtimes W_p = \{w \in C^\omega \mid \pi_\ell(w) \in W_\ell, \text{ where $\ell$ is maximal such that } \pi_\ell(w) \text{ is infinite}\}.
\]
Clearly, this operation is not commutative. 
We write $W_0 \ltimes W_1$ to denote $W_1\rtimes W_0$; we call it the
\emph{min\=/lexicographic product} of the objectives, for which more importance is given to $W_0$.
The difference will be important once we study infinite products.
We define infinite max\=/lexicographic products in the next section, and later consider (infinite)
min\=/lexicographic products in Section~\ref{sec:min-lex}.

In the rest of this section we discuss an associated operation: the
\emph{max\=/lexicographic product of two ordered graphs}
over disjoint sets of colours.
Given an~ordered $C_0$-graph $(G_0,\geq_0)$ and an~ordered $C_1$-graph $(G_1,\geq_1)$,
where $C_0 \cap C_1 = \emptyset$, we define $(G_0 \rtimes
G_1,\geq)$ to be the ordered $C_0\cup C_1$-graph with vertices  $V(G_0
\rtimes G_1) = V(G_0) \times V(G_1)$ ordered by
\[
    (v_0,v_1) \geq (v_0',v_1') \quad \iff \quad v_1 >_1 v_1' \tor [v_1 = v_1' \tand (v_0 \geq_0 v_0')]. 
\]
Its edges are
\[
    \begin{aligned}
    E(G_0 \rtimes G_1) = 
    &\, \{(v_0,v_1) \re {c_1} (v_0', v_1') \mid c_1 \in C_1 \tand&&\!\!\!\!\!\! v_1 \re {c_1} v_1' \in E(G_1)\} \ \cup \\
    &\, \{(v_0,v_1) \re {c_0} (v_0',v_1') \mid c_0 \in C_0 \tand&&\!\!\!\!\!\! [v_1 >_1 v_1' \tor\\
    &&& (v_1 = v_1' \tand v_0\re {c_0} v_0' \in E(G_0))]\}.
    \end{aligned}
\]
Once again, it is immediate to check that, if $G_0$ and $G_1$ are well-ordered or monotone then so is their lexicographic product. 
We write $G_0 \ltimes G_1$ to denote $G_1 \rtimes G_0$ with vertices written as $(v_0,v_1)$ instead of $(v_1,v_0)$, and call it the \emph{min\=/lexicographic product} of the graphs; as its edges do not depend on the order of $G_1$, it is still defined as a graph (though not as an ordered graph) when $G_1$ is not ordered.

Ohlmann\footnote{Formally, it was only proved for totally ordered graphs in~\cite{Ohlmann23}, but the proof for non-totally ordered graphs, presented in~\cite{CO25} for completeness, is the same.} related finite lexicographic products of positional objectives to lexicographic products of their universal graphs as follows.

\begin{thmC}[{\cite[Theorem~5.2]{Ohlmann23}}]\label{thm:universality_finite_lexico}
    Let $W_0 \subseteq C_0^\omega$, $W_1 \subseteq C_1^\omega$ be
    prefix-independent objectives with $C_0 \cap C_1 = \emptyset$. Let
    $\kappa$ be a cardinal, and assume that $U_0$ and $U_1$ are well-ordered monotone graphs which are
    $\kappa$-universal for $W_0$ and $W_1$, respectively.
    Then $U_0 \rtimes U_1$ is $\kappa$-universal for $W_0 \rtimes
    W_1$.
\end{thmC}

As a direct consequence, we get the following closure properties.

\begin{cor}
Prefix-independent  positional objectives are closed under finite lexicographic products.
\end{cor}

As an important example, the parity condition can be defined as the lexicographic product
\[
    \Parity_{d} = \TW_0 \rtimes \TL_1 \rtimes \TW_2 \rtimes \dots \rtimes \TL_{d-1} \rtimes \TW_d,
\]
where $d$ is an even integer.
Then, by Theorem~\ref{thm:universality_finite_lexico} and $\kappa$-universality of $\loopC{c}$ and $\bullet\caprodPow{c} \kappa$ for $\TW_c$ and $\TL_c$, respectively, we get that the graph
\[
    \loopC 0 \rtimes (\bullet \caprodPow{1}  \kappa) \rtimes  \loopC 2 \rtimes  \dots \rtimes  (\bullet \caprod^{d-1} \!\!\!\!\!\!\! \kappa) \rtimes  \loopC d 
\]
is $\kappa$-universal for $\Parity_d$.
A closer examination reveals that this graph corresponds to Walukiewicz's
signatures~\cite{Walukiewicz96}, or to Emerson and Jutla's positionality
proof~\cite{EJ91} (we also refer the reader to~\cite[Chapter~5]{Ohlmann21PhD} for discussions around this construction).

The purpose of this paper is to introduce extensions of finite lexicographic
products to infinite families of objectives, indexed by ordinals, and then to
give corresponding constructions over universal graphs in order to generalise
Theorem~\ref{thm:universality_finite_lexico} and obtain closure properties.
As we will see, in the infinite case max\=/lexicographic products and
min\=/lexicographic products behave quite differently. 
We treat them separately in Sections~\ref{sec:max-lex} and~\ref{sec:min-lex}.

\section{Infinite max-lexicographic products and topological completeness in the difference hierarchy}\label{sec:max-lex}

\subsection{Definitions and statement of the result.}

Fix a countable ordinal $\alpha$.
We fix a~family of pairwise disjoint sets of colours $(C_\lambda)_{\lambda <
\alpha}$ and a~family of prefix-independent objectives $(W_\lambda)_{\lambda <
\alpha}$ with $W_\lambda \subseteq C_\lambda^\omega$. %
We define $C=\bigcup_{\lambda < \alpha}C_\lambda$ and $C_{<\lambda},C_{\leq \lambda},C_{> \lambda},C_{\geq \lambda}$ as expected.

For a word $w \in C^\omega$, and an ordinal $\lambda< \alpha$, we let $\pi_\lambda(w) \in C_\lambda^* \cup C_\lambda^\omega$ denote the (finite or infinite) restriction of $w$ to colours in $C_\lambda$.
For a (finite or infinite) word $w=c_0c_1 \dots \in C^* \cup C^\omega$, we also
let $\ind(w)=\lambda_0 \lambda_1 \dots \in \alpha^* \cup
\alpha^\omega$ denote the (finite or infinite) word of ordinals
such that for all $i$ we have $c_i \in C_{\lambda_i}$.
Given $\Lambda=\lambda_0 \lambda_1 \dots \in \alpha^\omega$, recall that
\[
    \limsup \Lambda = \min_{i < \omega} \sup \{\lambda_i, \lambda_{i+1}, \dots \}.
\]
We note that $\limsup \Lambda$ is always defined and $\leq \a$, as it is a min of a set of ordinals $\leq \a$.
For notational convenience, we let $\pi_\alpha(w)=\emptyword$ and $W_\alpha = \emptyset$.

We define the max-lexicographic product of the family $(W_{\lambda})_{\lambda <  \alpha}$ to be
\[
    \maxLexProduct{\lambda<\a} W_\lambda = \{w \in C^\omega \mid \pi_{\lambda}(w) \in W_\lambda \text{ where }\lambda = \limsup \ind(w)\}.
\]
For a word $w$ to belong to the product, $\pi_\lambda(w)$ must be infinite, where $\lambda = \limsup \ind(w)$; in other words, the limsup of the indices must occur infinitely often.

Our main result in this section is the following.
\begin{thm}\label{thm:main_max_lex}
    Prefix-independent positional objectives are closed under countable max-lexicographic products.
\end{thm}

\subsection{Universal graph for max-lexicographic products.}
This subsection is devoted to the proof of Theorem~\ref{thm:main_max_lex}.

\paragraph{Colour-increasing unions}
We start by providing the following weakening of Kopczyński's conjecture, which will be the key lemma in the proof of Theorem~\ref{thm:main_max_lex} and may be of independent interest.

\begin{thm}\label{thm:weak_kopczynski}
    Let $(C_\lambda)_{\lambda < \alpha}$ be a family of sets of colours satisfying $C_\lambda \subseteq C_{\lambda'}$ for $\lambda<\lambda'<\alpha$, where $\alpha$ is countable.
    Let $(W_\lambda)_{\lambda < \alpha}$ be a family of prefix-independent positional objectives over the respective sets of colours such that for each $\lambda<\lambda'$ it holds that $C_\lambda^\omega \cap W_{\lambda'}= W_\lambda$.
    Then the union of the $W_\lambda$'s is positional.
\end{thm}

We will say that a family of objectives as above is colour-increasing.
The proof is a~simple application of Lemma~\ref{lemma:almost-universal}.

\begin{proof}
    Let $\kappa$ be a cardinal and let $U_0,U_1,\dots$ be well-ordered monotone $\kappa$\=/universal graphs for the respective objectives. Let $W$ be the union of all the $W_\lambda$'s.
    Let $U = \overset{\leftarrow}{\sum}_{\lambda<\alpha} U_{\lambda}$; we claim that $U$ is almost $(\kappa,W)$\=/universal and therefore $U\caprod \kappa$ is $(\kappa,W)$\=/universal.
    First, observe that $U$ indeed satisfies $W$: this follows from prefix\=/independence and the fact that each $U_\lambda$ satisfies $W_\lambda \subseteq W$.

    Now, consider a graph $G$ of size $<\kappa$ satisfying $W$.
    We should prove that for some $v$, $G[v] \to U$.
    We claim that there exists $v \in V(G)$ such that all colours appearing on paths from $v$ belong to $C_\lambda$ for some $\lambda$.
    Assume by contradiction that this fails. %
    Then, by an easy induction we obtain a path visiting edges with colours outside $C_{\lambda_0}$, $C_{\lambda_1},\dots$ where we choose $\lambda_0,\lambda_1,\dots$, to form a cofinal sequence
    in $\alpha$; such a path cannot satisfy any $W_\lambda$ and therefore it does not satisfy~$W$.

    We conclude that for some $v$ and some $\lambda$, it holds that $G[v]$ satisfies $W \cap C_\lambda^\omega = W_\lambda$.
    Therefore, $G[v] \to U_\lambda$, which concludes the proof since $U_\lambda\to U$.
\end{proof}

\paragraph{Proof of Theorem~\ref{thm:main_max_lex}}

For $\alpha'\leq\alpha$, we let $W_{<\alpha'}$ denote the max-lexicographic product of the family $(W_{\lambda})_{\lambda <  \alpha'}$ (over $C_{<\alpha'}$).
To prove Theorem~\ref{thm:main_max_lex}, we proceed by induction on $\alpha'$.
There are two cases, corresponding to $\alpha'$ being a successor or a limit.
First, we prove that for successor ordinals, our definition behaves 
just like finite lexicographic products.

\begin{lem}\label{lem:max_prod_non_limit}
For any $\alpha' < \alpha$, we have
\[
    W_{<\alpha'+1} = W_{<\alpha'} \rtimes W_{\alpha'}.
\]
\end{lem}

\begin{proof}
Let $w \in C_{< \alpha'+1}^\omega$.
\begin{itemize}
\item First assume that $\pi_{\alpha'}(w)$ is infinite.
Then $\limsup \ind(w)=\alpha'$ and we have
\[
    w \in W_{<\alpha'+1} \iff \pi_{\alpha'}(w) \in W_{\alpha'} \iff w \in W_{<\alpha'} \rtimes W_{\alpha'}\ .
\]
\item Otherwise, $\pi_{\alpha'}(w)$ is finite, and we let $w'$ denote a suffix of $w$ with $\pi_{\alpha'}(w')=\emptyword$.
Then we have
\[
    w \in W_{<\alpha'+1} \iff w' \in W_{<\alpha'+1} \iff w' \in W_{<\alpha'} \iff w \in W_{<\alpha'} \rtimes W_{\alpha'}. \qedhere
\]
\end{itemize}
\end{proof}

\noindent
On the other hand, for limit ordinals, our definition resembles a union.

\begin{lem}\label{lem:max_prod_limit}
For any limit ordinal $\alpha'\leq\alpha$, we have
\[
    W_{<\alpha'} = \bigcup_{\lambda<\alpha'} W_{<\lambda}.
\]
\end{lem}

\begin{proof}
It is a direct check that for $\lambda<\alpha'$ we have $W_{<\alpha'} \cap C_{<\lambda}^\omega = W_{<\lambda}$, and thus the right-to-left inclusion holds.
Conversely, let $w \in W_{<\alpha'}$.
Then $\lambda'=\limsup \ind(w)$ is $\leq \alpha'$ and $\pi_{\lambda'}(w)$ is infinite, so $\lambda'<\alpha'$.
Since only finitely many letters of $w$ have indices $>\lambda'$ and $\alpha'$ is a limit ordinal, we have $w \in C_{<\lambda}^\omega$ for some $\lambda<\alpha'$, and thus $w \in W_{<\alpha'} \cap C_{<\lambda}^\omega = W_{<\lambda}$.
\end{proof}

\noindent
Together, Lemmas~\ref{lem:max_prod_non_limit} and~\ref{lem:max_prod_limit} give an alternative inductive definition of the max-lexicographic product.
Now note that for any limit ordinal $\alpha'\leq\alpha$, the above union is colour-increasing: $(C_{<\lambda})_{\lambda<\alpha'}$ is an increasing sequence of sets of colours, $W_{<\lambda} \subseteq C_{<\lambda}^\omega$ and for any $\lambda<\lambda'$ we have $C_{<\lambda}^\omega \cap W_{<\lambda'} = W_{<\lambda}$.
Thus Theorem~\ref{thm:main_max_lex} holds by induction on $\alpha$: the successor case follows from Lemma~\ref{lem:max_prod_non_limit} and Theorem~\ref{thm:universality_finite_lexico} and the limit case follows from Lemma~\ref{lem:max_prod_limit} and Theorem~\ref{thm:weak_kopczynski}.

\subsection{Max-Parity: Positionality and topological completeness.}\label{subsec:maxparity}
We now discuss the important case of the Max-Parity languages.
Let $C_\lambda=\{\lambda\}$ for $\lambda <  \alpha$ and\footnote{We recall that the parity of an ordinal $\alpha$ is the parity of the unique $n<\omega$ such that $\alpha$ can be written as $\alpha' + n$ for $\alpha'$ either $0$ or a~limit ordinal.} %
\[
    W_\lambda=\begin{cases} \TL_\lambda \tif \lambda \text{ is even,} \\ \TW_\lambda \tow. \end{cases}
\]

We define the Max-Parity objective $\MaxParity_\alpha$ as the lexicographic product of the $W_\lambda$'s for $\lambda<\alpha$.
Equivalently, it can be written as
\[
    \MaxParity_\alpha = \{w \in \alpha ^\omega \mid \limsup w \text{ is odd}\}.
\]
The following remark justifies our choice of odd priorities to be winning, rather than the more standard even ones.

\begin{rem}
\label{rem:why-odd}
Note that if $\lambda = \limsup w$ is an~odd ordinal, it is necessarily non-limit, and therefore $\pi_{\lambda}(w)$ is infinite.
Moreover, words whose $\limsup$ is a limit ordinal must be rejected: accepting them would break positionality, as producing an increasing sequence of priorities requires memory (compare with Figure~\ref{fig:counter_example_intro}).
\end{rem}

\begin{cor}
    For every countable ordinal $\alpha$, $\MaxParity_\alpha$ is positional.
\end{cor}

\paragraph{Completeness in the difference hierarchy.}
We first recall the definition of the difference hierarchy (see also~\cite[Chapter~22.E]{Kechris1995}).
For a sequence of sets $(A_\eta)_{\eta<\alpha}$, $A_\eta\subseteq C^\omega$, we define $D_\alpha\big((A_\eta)_{\eta<\alpha}\big)$ as the set containing the elements $w\in \bigcup_{\eta<\alpha} A_\eta$ where the least $\eta<\alpha$ such that $w\in A_\eta$ has parity opposite to that of $\alpha$.
The class $D_\alpha(\bsigma{2})$ consists of the sets that can be described as $D_\alpha\big((A_\eta)_{\eta<\alpha}\big)$ for an increasing sequence of $\bsigma{2}$-sets. 
The class $D_\alpha(\bpi{2})$ consists of the sets that are complements of $D_\alpha(\bsigma{2})$-sets.
(Note that $D_1(\bsigma{2}) = \bsigma{2}$ and $D_1(\bpi{2}) = \bpi{2}$.)

In the remainder of the section, we show that the languages $\MaxParity_\alpha$ are complete (with respect to continuous reductions) for infinite levels of the difference hierarchy.

Recall that a function $h\colon C^\omega \to X^\omega$ is continuous if and only if the value of the $n$th letter of $h(w)$ only depends on a finite prefix of $w$; we refer to~\cite{Kechris1995} for a formal definition. Since our objectives in $X^\omega$ admit neutral letters, we will assume that the functions are \emph{$1$\=/Lipschitz}, i.e.~the $n$th letter of $h(w)$ depends only on the first $n$ letters of $w$. Such functions can be represented by $f\colon C^\ast \to X$ with $h=\tilde{f}\colon C^\omega\to X^\omega$ defined as $\big(\tilde{f}(w)\big)_n = f(w_0w_1\cdots w_{n-1})$.

\begin{thm}
    \label{thm:difference-complete}
    For each even $\alpha<\omega_1$, the language $\MaxParity_{\alpha+1}$ is complete for $D_\alpha(\bsigma{2})$.
    For each odd $\alpha<\omega_1$, the language $\MaxParity_{\alpha+1}$ is complete for $D_\alpha(\bpi{2})$.
    For each limit ordinal $\alpha<\omega_1$, the language $\MaxParity_{\alpha}$ is complete for $D_\alpha(\bsigma{2})$.
\end{thm}

\noindent
Before proving Theorem~\ref{thm:difference-complete}, we need an auxiliary result (Lemma~\ref{lem:decreasing-reductions}) that might be of independent interest.

We let  $\coBuchi \subseteq \{1,2\}^\omega$ be the language of words in which $2$ appears finitely often ($1$ serves as a~neutral letter in this language). We recall that $\coBuchi$ is complete for $\bsigma{2}$, that is, for every $\bsigma{2}$-set $A\subseteq C^\omega$ there is a continuous function $h\colon C^\omega \to \{1,2\}^\omega$ such that $h^{-1}(\coBuchi)=A$.
As the language $\coBuchi$ admits a~neutral letter, we can assume that such a reduction $h$ is induced by a function $f\colon C^\ast\to\{1,2\}$ such that $h=\tilde{f}$. In this case we say that $f$ \emph{represents} the reduction $h$.

Given a~function $f\colon C^\ast \to \{1,2\}$ we define $\preimCB{f}$ as the set of words $w\in C^\omega$ such that $f(w_0w_1\cdots w_{n-1})=2$ for only finitely many $n$. In other words $\preimCB{f}=\tilde{f}^{-1}(\coBuchi)$.

Let $f,g\colon C^\ast \to \{1,2\}$.
We write $f\leq g$ if for all $x\in C^\ast$, $g(x) = 1$ implies $f(x) = 1$ (i.e.~$f(x)\leq g(x)$).
Note that if $f\leq g$, then $\preimCB{f} \supseteq \preimCB{g}$.
We say that a sequence of functions $(f_\eta)_{\eta<\alpha}$, with $f_\eta\colon C^\ast \to \{1,2\}$, is \emph{pointwise decreasing} if $f_{\eta} \geq f_{\eta'}$ for $\eta \leq \eta'<\alpha$.

\begin{lem}\label{lem:decreasing-reductions}
    Let $\alpha < \omega_1$ and let $(A_\eta)_{\eta< \alpha}$ be an increasing sequence of $\bsigma{2}$-subsets of $C^\omega$.
    Then, there exists a pointwise decreasing sequence of functions $(f_\eta)_{\eta<\alpha}$ such that $f_\eta\colon C^\ast \to \{1,2\}$ is a~representation of a~reduction of $A_\eta$ to $\coBuchi$ (that is, $\preimCB{f_\eta} = A_\eta$). 
\end{lem}
\begin{proof}
    We prove the following stronger statement for all $\alpha<\omega_1$ by transfinite induction:
    \begin{equation}\tag{$\star$}
    \begin{gathered}
         \text{For any two functions } \gbig\leq \gsmall  \text{ and every increasing sequence } (A_\eta)_{\eta< \alpha} \text{ of } \bsigma{2}\text{-sets}\\
          \text{such that } \preimCB{\gsmall} \subseteq A_\eta \subseteq \preimCB{\gbig} \quad \text{ for all } \eta<\alpha, \\
         \text{there is a pointwise decreasing sequence of functions }  (f_\eta)_{\eta<\alpha} \text{ such that:} \\
         \text{for all } \eta\leq \eta'<\alpha: \quad \preimCB{f_\eta} = A_\eta \quad \tand \quad \gbig \leq f_{\eta'} \leq f_{\eta} \leq \gsmall.
    \end{gathered} \label{eq:property-star}        
\end{equation}

    First, we introduce two operators performing union and intersection of representations.
    Let $f,g \colon C^\ast \to \{1,2\}$ be two functions.
    Consider their pointwise maximum $h = \max(f,g)\colon C^\ast\to \{1,2\}$. Then $f,g \leq h$ and $\preimCB{h} = \preimCB{f}\cap \preimCB{g}$ (the maximum contains finitely many $2$s if and only if both of them do). In particular, if $\preimCB{f} \subseteq \preimCB{g}$, then $\preimCB{h} = \preimCB{f}$.

    We define a function $\union(f,g)\colon C^\ast \to \{1,2\}$. For $x\in C^\ast$, let $x_g$ be the longest non\=/strict prefix of $x$ such that $g(x_g) = 2$, and let $x_f'$ be the longest strict prefix of $x$ such that $f(x'_f) = 2$ (both can be empty if there is no such prefix). We define: 
    \[ \union(f,g) = \begin{cases}
        1 \tif f(x) = 1,\\
        1 \tif f(x) = 2 \, \tand \, |x_g| \leq |x'_f|,\\
        2 \tif f(x) = 2 \, \tand \,  |x_g| > |x'_f|.
    \end{cases}\]

    Note that the function $\union(f,g)$ is defined in an~asymmetric way and $\union(f,g) \leq f$. We claim that $\preimCB{\union(f,g)} = \preimCB{f}\cup \preimCB{g}$, in particular, $\preimCB{\union(f,g)} = \preimCB{g}$ if $\preimCB{f} \subseteq \preimCB{g}$.
    The inclusion $\preimCB{f} \subseteq \preimCB{\union(f,g)}$ follows from the inequality $\union(f,g) \leq f$.
    Let $w\in \preimCB{g}$. Since $\tilde{g}(w)$ eventually only contains $1$s, there is a constant $k$ such that for every prefix $x$ of $w$, $|x_g|\leq k$. Therefore, for sufficiently long prefixes, we are always in one of the first two cases.
    Finally, let $w\notin \preimCB{f} \cup \preimCB{g}$; we build an infinite sequence of prefixes $(x_i)_{i>0}$ of $w$ such that $\union(f,g)(x_i) =2$. Assume that $x_i$ has been built and 
    let $x_i'$ be a prefix of $w$ of length $>|x_i|$ such that $g(x_i') = 2$. Let $x_{i+1}$ be the first extension of $x_i'$ such that $f(x_{i+1}) = 2$ (possibly $x_{i+1} = x_i'$).
    By definition of $\union$, we have $\union(f,g)(x_{i+1}) =2$, as desired.

    We show statement~\eqref{eq:property-star} by induction on $\alpha$.
    Let $\alpha = \alpha' +1$ be a successor ordinal.
    Let  $f'_{\alpha'}$ be a~representation of a~reduction of $A_{\alpha'}$ to $\coBuchi$, and let
        \[ f_{\alpha'} = \max(\union(\gsmall, f'_{\alpha'}),\gbig).\]
    
    By the previous remarks, we have that $\gbig\leq f_{\alpha'} \leq \gsmall$.
Indeed, $\union(\gsmall, f'_{\alpha'}) \leq  \gsmall$ by the property of $\union$, and since $\gbig \leq \gsmall$, their maximum is also smaller than $\gsmall$ so $f_{\alpha'} \leq \gsmall$; the fact that  $\gbig\leq f_{\alpha'}$ is straightforward.

Moreover $\preimCB{f_{\alpha'}} = A_{\alpha'}$ since
\[
    \preimCB{f_{\alpha'}} = ( \preimCB{\gsmall} \cup \preimCB{f'_{\alpha'}} ) \cap \preimCB{\gbig} \quad  \text{ and } \quad   \preimCB{\gsmall}  \subseteq \preimCB{f'_{\alpha'}}  \subseteq \preimCB{\gbig} . 
\]
Now, apply the induction hypothesis~\eqref{eq:property-star} to $f_{\alpha'} \leq \gsmall$ and the sequence $(A_\eta)_{\eta<\alpha'}$.
The obtained sequence, together with $f_{\alpha'}$, is as desired.

    Now let $\alpha$ be a limit ordinal. Let $0=\alpha_0 <\alpha_1 < \dots$ be a sequence of ordinals $<\alpha$ with $\alpha = \sup \alpha_i$.
    Let $f'_i$ be a~representation of any reduction from $A_{\alpha_i}$ to $\coBuchi$. We define by induction:
    \[ f_0 = \max(\union(\gsmall, f'_0), \gbig) \quad \tand \quad f_{i+1} = \max(\union(f_i, f'_{i+1}), \gbig). \]
    In this way, we obtain that $\gbig \leq \dots \leq f_2 \leq f_1 \leq f_0 \leq \gsmall$ and $\preimCB{f_i} = A_{\alpha_i}$.
    Now, for each $i$ apply the induction hypothesis with functions $f_{i+1} \leq f_i$ and the sequence $(A_\eta)_{\alpha_i < \eta <\alpha_{i+1}}$ (note that this sequence has order type $<\alpha$).
    The concatenation of the obtained sequences, with $f_i$ placed at position $\alpha_i$, is as desired.
\end{proof}

\noindent
We are now ready to prove Theorem~\ref{thm:difference-complete}.

\begin{proof}[Proof of Theorem~\ref{thm:difference-complete}]
    We begin by showing that $\MaxParity_{\alpha+1}$ belongs to $ D_\alpha(\bsigma{2})$ if $\alpha$ is even, and to $ D_\alpha(\bpi{2})$ if $\alpha$ is odd.
    Let $A_\eta$ be the set of sequences $w\in (\alpha+1)^\omega$ where elements strictly larger than $\eta$ appear only finitely many times (equivalently, $\limsup w \leq \eta$).
    These sets are clearly in $\bsigma{2}$.
    Let $w\in (\alpha+1)^\omega$. We have:
    \[w\in D_\alpha\big((A_\eta)_{\eta<\alpha}\big) \quad \Longleftrightarrow \quad \limsup w \text{ is } <\alpha \text{ and has parity opposite to } \alpha.\]
    Therefore, for even $\alpha$, $\MaxParity_{\alpha+1}$ equals $D_\alpha\big((A_\eta)_{\eta<\alpha}\big)$, and for odd $\alpha$, $\MaxParity_{\alpha+1}$ equals the complement of $D_\alpha\big((A_\eta)_{\eta<\alpha}\big)$.

    For the opposite direction, we show how to reduce an arbitrary set in $D_\alpha(\bsigma{2})$ (respectively, in $D_\alpha(\bpi{2})$) to $\MaxParity_{\alpha+1}$ if $\alpha$ is even (respectively, if $\alpha$ is odd). We prove the third statement about limit ordinals at the end.

    Assume $\alpha$ is even and let $X=D_\alpha\big((A_\eta)_{\eta<\alpha}\big)$, for $(A_{\eta})_{\eta < \alpha}$ an increasing sequence of $\bsigma{2}$-sets in some space $C^\omega$. 
    By Lemma~\ref{lem:decreasing-reductions}, there is a pointwise decreasing sequence of functions $(f_\eta)_{\eta<\alpha}$ such that $f_\eta\colon C^\ast \to \{1,2\}$ and $\preimCB{f_\eta} = A_\eta$.
    We define a representation $f\colon C^\ast \to (\alpha+1)$ of a~reduction of $X$ to $\MaxParity_{\alpha+1}$ by: 
    \[f(x) = \begin{cases}
      \alpha & \tif f_\eta(x) = 2 \text{ for all } \eta,\\
     \inf \{\eta < \alpha \mid f_{\eta}(x) = 1\} & \text{ if not}.
    \end{cases}\]   
       
    It remains to see that for $w\in C^\omega$, $\tilde{f}(w)\in \MaxParity_{\alpha+1}$ if and only if $w\in X$.
    Let $\tilde{f}(w) = \eta_1\eta_2\eta_3\dots$, and let $\eta_{\sup} = \limsup_{i<\omega} \eta_i$. 
    First, if $\eta_{\sup} = \alpha$ (even), then $w\notin \preimCB{f_\eta} =A_\eta$ for any $\eta$, so $w\notin X$.
    If $\eta_{\sup}<\alpha$, then $\tilde{f_{\eta_{\sup}}}(w)$ eventually only contains $1$s, since the sequence $(f_\eta)_{\eta<\alpha}$ is decreasing.
    Therefore, $w\in \preimCB{f_{\eta_{\sup}}} = A_{\eta_{\sup}}$.
    Also, for $\eta<\eta_{\sup}$, the sequence $\tilde{f_\eta}(w)$ contains infinitely many $2$s, at positions $i$ where $\eta_i > \eta$.
    We conclude that $\eta_{\sup}$ is the smallest ordinal such that $w\in A_{\eta_{\sup}}$, so:
    \[ w\in X \quad \Longleftrightarrow \quad \eta_{\sup} \text{ is odd } \quad \Longleftrightarrow \quad \tilde{f}(w)\in \MaxParity_{\alpha+1},\]
    as desired.

    For the case $\alpha$ odd, let $X$ be a set such that its complement equals some $D_\alpha\big((A_\eta)_{\eta<\alpha}\big)$ for an increasing sequence of $\bsigma{2}$-sets. Defining $f$ and $\eta_{\sup}$ as before, we obtain:
    \[ w\notin X \quad \Longleftrightarrow \quad \eta_{\sup} \text{ is even } \quad \Longleftrightarrow \quad \tilde{f}(w)\notin \MaxParity_{\alpha+1},\]
    showing that $\MaxParity_{\alpha+1}$ is complete for $D_\alpha(\bpi{2})$.

    Finally, we prove that $\MaxParity_{\alpha}$ is complete for $D_\alpha(\bsigma{2})$ when $\alpha$ is a limit ordinal.
    It suffices to show that $\MaxParity_{\alpha +1}$ reduces to $\MaxParity_{\alpha}$, as we have already proved $D_\alpha(\bsigma{2})$-hardness of the former.
    Let $(\gamma_i)_{i<\omega}$ be an increasing sequence of ordinals $<\alpha$ such that $\alpha = \sup \gamma_i$.
    It suffices to consider the function $f\colon (\alpha+1)^\omega \to \alpha^\omega$ that replaces occurrences of $\alpha$ at position $i$ by $\gamma_i$.
    
    This concludes the proof of Theorem~\ref{thm:difference-complete}.
    \end{proof}

\noindent
As far as we are aware, this constitutes the first positionality proof for complete languages for infinite levels of the difference hierarchy.
This lays the foundation for the systematic study of positional objectives within $\bdelta{3}$, the natural topological generalisation of $\omega$-regular objectives.

\section{Infinite min-lexicographic products}\label{sec:min-lex}

We introduce infinite min-lexicographic products of a sequence of winning objectives.
Intuitively, the objectives at the beginning of the sequence have priority over those that appear later. 
The sequence can be indexed by any ordinal.
We show that if winning conditions in the product are positional then the
min-lexicographic product objective is positional too (Theorem~\ref{thm:universality_min_lexico}). 
For this we provide an adequate construction of a universal graph from universal
graphs for the components. 

As we shall see, min-lexicographic products turn out to be more complex than max\=/lexicographic ones, for different aspects listed below.
\begin{itemize}
\item Finding the natural definition of infinite min-lexicographic products indexed over ordinals~${>\omega}$ is not obvious (see also Remark~\ref{rem:def_min_prod} below for more explanations). 
\item Topologically, min-lexicographic products generally lie beyond $\bdelta{3}$ (for instance, the product of $\omega$-many trivially winning conditions is in fact $\bsigma{3}$-complete).
\item Constructions (and universality proofs) establishing their positionality turn out to be substantially more involved. 
\end{itemize}

\subsection{Definitions and statement of the result}

\paragraph{Setting.} In this section, we fix a cardinal $\kappa \geq 2$, an
ordinal $\alpha$, a family of pairwise disjoint sets of colours
$(C_\lambda)_{\lambda < \alpha}$, and a family of prefix-independent objectives
$(W_\lambda)_{\lambda < \alpha}$ with $W_\lambda \subseteq C_\lambda^\omega$ for
all $\lambda$.
We assume that each $W_\lambda$ has a $\k$-universal well-ordered monotone graph
$(U_\lambda,\geq_\lambda)_{\lambda < \a}$.
We will use $C = \bigcup_{\lambda<\alpha} C_\lambda$, as well as $C_{<\lambda},C_{\leq
\lambda},C_{>\lambda},C_{\geq \lambda}$ defined as expected.
For a word $w \in C^\omega$, and an ordinal $\lambda<\alpha$, we let
$\pi_\lambda(w) \in C_\lambda^* \cup C_\lambda^\omega$ denote the (finite or
infinite) projection of $w$ to colours in $C_\lambda$.
Likewise, we let $\pi_{<\lambda}(w)$ denote the projection of $w$ to colours in $C_{<\lambda}$.

\paragraph{Min-lexicographic products.} We say that a word $w$ is
\emph{$\lambda$-supported}\footnote{Formally, being $\lambda$-supported depends
on the sequence $(C_\lambda)_{\lambda < \alpha}$. We do not
explicitly include this dependence in the notation for simplicity.} if
$\pi_\lambda(w)$ is infinite and
$\pi_{<\lambda}(w)$ is finite.
A word is \emph{supported} if it is $\lambda$-supported for some $\lambda$. 
In other words, a word is $\lambda$-supported if (i) after a finite prefix, $\lambda$ is the
smallest index of colours that appears, and (ii) a colour from $C_\lambda$ appears infinitely often. 
In particular, $\lambda$ is uniquely determined by $w$. 
For example, if $\alpha=\omega+1$ and $C_{\lambda}=\{\lambda\}$, then the word
$0\omega1\omega2\omega \dots \in C^\omega$ is not supported, and the word
$12131415\dots$ is $1$-supported. 

We define the \emph{min-lexicographic product} of $(W_\lambda)_{\lambda < \alpha}$ to be
\[
    \minLexProduct{\lambda<\a} W_\lambda = \{w \in C^\omega \mid w \text{ is $\lambda$-supported and } \pi_\lambda(w) \in W_\lambda\}.
\]
Note that for $\alpha < \omega$, every word is supported, and thus in
this case our definition coincides with finite min-lexicographic products. 

\begin{lem}
	The min-lexicographic product is associative. Formally, let $(\m_i)_{i\leq\beta}$ be a strictly increasing sequence of ordinals with $\m_0=0$, $\m_\beta=\a$, and $\m_i=\sup_{j<i}\m_j$ for every limit ordinal $i\leq\beta$. Then
	\[ \minLexProduct{\lambda<\a} W_\lambda = \minLexProduct{i<\beta} \left( \minLexProduct{\mu_i\leq \lambda < \mu_{i+1}} W_\lambda \right). \]
\end{lem}
\begin{proof}
	Let $W$ be the objective on the left of the equality and $\widetilde{W}$ the one on the right.
	Assume that $w\in W$. Then, $w$ is $\lambda$-supported for some $\lambda<\a$ (for the $\a$-partition of the set of colours) and $\pi_\lambda(w)\in W_\lambda$. Let $i$ be the unique ordinal $<\beta$ such that $\m_i\leq \lambda < \m_{i+1}$. Then, $w$ is $i$-supported (for the $\beta$-partition of the set of colours), and $\pi_i(w)$ is, in turn, $\lambda$-supported and $\pi_\lambda(\pi_i(w)) = \pi_\lambda(w)\in W_\lambda$, so $\pi_i(w) \in \minLexProduct{\mu_i\leq \lambda < \mu_{i+1}} W_\lambda$ and $w\in \widetilde{W}$.
	
	Conversely, assume $w\notin W$. If $w$ is supported, we conclude using the same argument as above. Assume $w$ is not supported (for the $\a$-partition of the set of colours). Then, for every $i<\beta$ such that $w$ is $i$-supported (for the $\beta$-partition of the set of colours), $\pi_i(w)$ is not $\lambda$-supported for any $\mu_i\leq \lambda < \mu_{i+1}$ (otherwise $w$ would be $\lambda$-supported), so $\pi_i(w)\notin \minLexProduct{\mu_i\leq \lambda < \mu_{i+1}} W_\lambda$, so $w\notin \widetilde{W}$.
\end{proof}

\begin{rem}\label{rem:def_min_prod}
Another possible definition of min-lexicographic product could be
\[
    W' = \{w \in C^\omega \mid \lambda_0= \mininf (w) \text{ is defined and }  \pi_{\lambda_0}(w) \in W_{\lambda_0}\},
\]
where $\mininf(w)$ is the minimal $\lambda<\alpha$ such that there are
infinitely many occurrences of colours from $C_\lambda$ in $w$.
The two definitions coincide for $\alpha \leq \omega$, but they are different
for $\a=\w+1$.
Indeed, take $C_\lambda=\{\lambda\}$, $W_i
= \TL_i$  for $i< \omega$ and $W_\omega=\TW_{\{\omega\}}$ (we write $\TW_{\{\omega\}}$ instead of
$\TW_{\omega}$ to avoid any ambiguities here).
Observe that $\mininf(0\w1\w2\w\dots) = \w$ while
this word is not supported. 
So this word is not in the min-lexicographic product, but it is in $W'$, showing that the two definitions
are different.

However, this modified definition has several disadvantages that already appear
for the example above.
Firstly, the modified operation is not associative. 
Indeed, the product of the $W_i$'s for $i<\omega$ is exactly
$\TL_{\omega}$, the trivially losing objective over $\omega$ (for both definitions). 
Hence, $0\w1\w2\w\dots\notin \TL_\omega \ltimes \TW_{\{\omega\}}$, so $\TL_\omega \ltimes \TW_{\{\omega\}} \neq W'$.

Moreover, $W'$ is not even positional: to win in the game from
Figure~\ref{fig:counter_example_naive_def}, Eve cannot use a positional
strategy.  
So the modified definition does not preserve positionality. 
As we will show, our definition does preserve positionality.

\begin{figure}[ht]
\begin{center}
\includegraphics[width=0.26\linewidth]{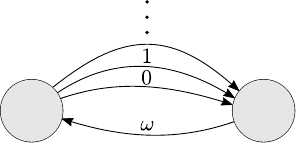}
\end{center}
\caption{A game in which Eve requires memory to ensure objective $W'$, for instance by playing a path labelled $0\omega 1 \omega \dots$.}\label{fig:counter_example_naive_def}
\end{figure}
\end{rem}

\paragraph{Main result.} We can now state the main result of the section: the closure of  positional objectives under infinite min-lexicographic products.

\begin{thm}\label{th:main-pos-min-lex}
	Prefix-independent  positional objectives are closed under arbitrary min-lexicographic products.
\end{thm}

\noindent
To prove Theorem~\ref{th:main-pos-min-lex}, we will construct a universal graph for the min-lexicographic product.

\subsection{Applications: \texorpdfstring{$\omega$}{w}-Büchi and Min-Parity.} Before moving on to the proof of Theorem~\ref{th:main-pos-min-lex} in Subsection~\ref{subsec:universal-graph-for-min-lex}, which spans most of the remaining section, we present two applications.

\paragraph{$\omega$-B\"uchi.}
For $\alpha=\omega$, $C_i=\{i\}$ and $W_i=\TW_i$ for $i<\alpha$, the min-lexicographic product yields:
\[
\omegaBuchi = \{w \in \omega^\omega \mid \exists i, |w|_i \text{ is infinite}\},
\]
which one can see as an infinite union of Büchi objectives.
Grädel and Walukiewicz~\cite{GW06} proved bi-positionality of $\omegaBuchi$ over vertex-labelled game graphs.
Theorem~\ref{th:main-pos-min-lex} implies positionality\footnote{We recall here that in this paper, positionality means ``positionality in the presence of a neutral letter''. For this objective, the fact that adding a neutral letter retains positionality is non-trivial.} over
edge-labelled game graphs; it is easy to see that positionality for the opponent
fails for edge-labelled graphs. 

We note that the language $\omegaBuchi$ is complete for the class $\bsigma{3}$~\cite[Exercise~23.2]{Kechris1995}.

\paragraph{Min-Parity.}
We now discuss the case of the Min-Parity languages.
We let $W_\lambda$ be the language over $C_\lambda=\{\lambda\}$ that equals $\TW_\lambda$ if $\lambda$ is even and $\TL_\lambda$ if $\lambda$ is odd.
We define the Min-Parity objective $\MinParity_\alpha$ as the lexicographic product of the $W_\lambda$'s for $\lambda<\alpha$.
Equivalently, it can be written as:
\[
    \MinParity_\alpha = \{w \in \alpha ^\omega \mid  w \text{ is $\lambda$-supported for an even $\lambda$}\}.
\]

\begin{cor}
    For every ordinal $\alpha$, $\MinParity_\alpha$ is positional.
\end{cor}

\noindent
For finite $\alpha$, $\MinParity_\alpha$ is complete for the finite levels of the difference hierarchy over $\bsigma{2}$, as it is interreducible with $\MaxParity_\alpha$,\footnote{More precisely, for finite $\alpha$, $\MinParity_\alpha$ is interreducible with $\MaxParity_\alpha$ for $\alpha$ even and with the complement of $\MaxParity_\alpha$ for $\alpha$ odd.}
whose completeness was shown in Theorem~\ref{thm:difference-complete} (see also~\cite{mskrzypczak_colorings}).
For infinite $\alpha$, however, $\MinParity_\alpha$ lies beyond $\bdelta{3}$.

\begin{thm}
    \label{thm:sigma3-complete}
    For each infinite countable $\alpha$, the language $\MinParity_{\alpha}$ is complete for $\bsigma{3}$.
\end{thm}
\begin{proof}
    To show that $\MinParity_{\alpha}$ belongs to $\bsigma{3}$, we can write it as:
    \[ \MinParity_\alpha = \bigcup_{\substack{\lambda<\alpha,\\ \lambda \text{ even}}} S_\lambda \quad \text{ , where } \;\; S_\lambda = \{ w\in \alpha^\omega \mid w \text{ is $\lambda$-supported}\}. \]
    Now, we can write each $S_\lambda$ as an intersection of a $\bsigma{2}$ and a $\bpi{2}$-set:
    \[ S_\lambda =  \{w\in \alpha^\omega \mid |w|_{<\lambda} \text{ is finite}\} \cap \{w\in \alpha^\omega \mid |w|_\lambda \text{ is infinite}\}. \] 
    In particular, $S_\lambda$ is a $\bsigma{3}$-set. Since $\bsigma{3}$ is closed under countable unions, we conclude that $\MinParity_{\alpha}$ belongs to $\bsigma{3}$.

    To show $\bsigma{3}$-hardness, we reduce $\omegaBuchi$ to $\MinParity_\alpha$, for $\alpha\geq \omega$. It suffices to take the function $f\colon \omega^\omega \to \alpha^\omega$ that replaces a letter $i$ by $2i$.
    In this way, $f(w)\in \MinParity_\alpha$ if and only if there is a number $i$ that appears infinitely often in $w$.
\end{proof}

\subsection{Universal graph for min-lexicographic products}\label{subsec:universal-graph-for-min-lex}
In the rest of the section, we let  $W = \minLexProduct{\lambda<\a} W_\lambda$.
To show positionality of $W$, we define for every ordinal $\b$ the \emph{power graph} $U^{\b}$, using the universal graphs $(U_\lambda,\geq_\lambda)_{\lambda < \a}$.
We show that $U^\b$ is $\k$-universal for $W$ if $\b$ is chosen large
enough (Theorem~\ref{thm:universality_min_lexico}).
Throughout this section, $\b$ is an arbitrary but fixed ordinal.

\subsubsection{Construction of the universal graph}
For each $\lambda$, we consider the ordered graph $U_\lambda^\top$ obtained from
$U_\lambda$ by adding a fresh maximal vertex $\top_\lambda$ with no incoming
edge and all possible outgoing edges except towards itself; formally,
$E(U_\lambda^\top) = E(U_\lambda) \cup (\{\top_\lambda\} \times C_\lambda \times
V(U_\lambda))$.
Note that $U_\lambda^\top$ is well-ordered, monotone, and $\kappa$-universal for $W_\lambda$.

\paragraph{Vertices of $U^{\beta}$.} The vertices of $U^{\beta}$ are the
pairs $(f,S)$, where $f\colon\alpha \to \beta$ is a non-increasing function and
$S\colon \alpha\to \bigcup_{\lambda<\alpha} V(U_\lambda^\top)$ is such that $S(\lambda)\in
V(U^\top_\lambda)$ for all $\lambda<\a$.
Moreover, the two functions are linked by the condition: For all $\lambda<\alpha$ (with the convention $f(\alpha)=0$),
\begin{equation}\label{eq:vertex-condition}
     S(\lambda) \neq \top_\lambda \quad \implies \quad f(\lambda) > f(\lambda+1).
\end{equation}
This condition implies that there may be only finitely many $\lambda$'s for which $S(\lambda)\neq \top_\lambda$.

\paragraph{Order over $U^\beta$.} 
A vertex $(f,S)$ is a pair of sequences. 
Vertices are ordered by lexicographic order over the interleaving of these two
sequences:  $f(0), S(0), f(1), S(1) \dots$ where lesser coordinates matter the most.
To define this formally we introduce a piece of notation.
Given $(f,S) \in V(U^{\beta})$ and
$\lambda \leq \alpha$, we let $(f,S)_{<\lambda}$ be obtained by restricting the domains of $f$ and $S$ to
$\lambda$.
We let $(f,S) > (f',S')$ if and only if
\[
    \exists \lambda < \alpha, \quad (f,S)_{<\lambda} = (f',S')_{<\lambda} \tand [f(\lambda)>f'(\lambda) \tor 
    (f(\lambda)=f'(\lambda) \tand S(\lambda) >_\lambda S'(\lambda))].
\]
Clearly, the above order is total assuming each $\geq_\lambda$ is.

It is also convenient to define $(f,S)_{\leb \lambda}$ that is  obtained from
$(f,S)$ by restricting $f$ to $\lambda+1$ and $S$ to $\lambda$ (equivalently,
$(f,S)_{\leb \lambda}$ is obtained by extending the map from $(f,S)_{<\lambda}$
by $\lambda \mapsto f(\lambda)$). 

Using this notation, we get that $(f,S)>(f',S')$ if and only if there exists $\lambda < \alpha$ such that
\[
  [(f,S)_{<\lambda} = (f',S')_{< \lambda} \tand f(\lambda) > f'(\lambda)] \quad \tor \quad [(f,S)_{\leb \lambda} = (f',S')_{\leb \lambda} \tand S(\lambda) >_\lambda S'(\lambda)].
\]

\paragraph{Edges of $U^\beta$.} For a colour $c_\lambda \in C_\lambda$ and
vertices $(f,S),(f',S') \in V(U^\beta)$, we let $(f,S) \re{c_\lambda} (f',S')
\in E(U^\beta)$ if and only if 
\[
    (f,S)_{\leb \lambda} > (f',S')_{\leb \lambda} \quad \tor \quad [(f,S)_{\leb \lambda} = (f',S')_{\leb \lambda} \tand S(\lambda) \re {c_\lambda} S'(\lambda) \in E(U_\lambda^\top)].
\]
This definition ensures a property we will often use in proofs:
\begin{equation}\label{eq:edges-keep-order}
    \text{if $(f,S) \re {c_\lambda} (f',S')$ and $c_\lambda\in C_\lambda$ then $(f,S)_{\leb \lambda} \geq (f',S')_{\leb \lambda}$.}
\end{equation}
This means that a transition on a colour $c_\lambda$ does not increase the part of
the state before coordinate $\lambda$, nor the $f$ component of the coordinate $\lambda$.

\pagebreak

\subsubsection{Monotonicity and compositionality}\label{subsec:universal-graph-monotonicity}

\paragraph{Monotonicity and satisfiability of $W$.} We show that the power graph $U^\b$ is monotone and satisfies $W$.

\begin{lem}\label{lem:prop_U}
    The graph $U^\b$ defined above has the following properties:
\begin{enumerate}[1.]
    \item it is well-ordered;
    \item it is monotone;
    \item it satisfies $W$.
\end{enumerate}
\end{lem}

\begin{proof} We prove the three items in order.
\begin{enumerate}[1.]
    \item Totality, reflexivity and transitivity are immediate, so we focus on well-foundedness.
    For contradiction, consider an infinite decreasing sequence $(f^i,S^i)_{i \in \omega}$ of vertices of $U^\beta$.
    Let $\lambda_0$ be the minimal $\lambda < \alpha$ such that $(f^i(\lambda),S^i(\lambda))$ is not constant.
    Since $(f^i,S^i)_{<\lambda_0}$ is constant, it must be that
    $(f^i(\lambda_0),S^i(\lambda_0))_{i\in\w}$ is non-increasing. 
    Now since $\beta \times V(U_{\lambda_0}^\top)$ is well-founded, the above
    sequence is ultimately constant.
    Let $i_0$ be the last index of strict decrease: $(f^{i_0}(\lambda_0),S^{i_0}(\lambda_0)) >
    (f^{i_0+1}(\lambda_0),S^{i_0+1}(\lambda_0)) =
    (f^{i}(\lambda_0),S^{i}(\lambda_0))$, for all $i > i_0$. 
    
    We show that $f^{i_0}(\lambda_0) > f^{i_0+1}(\lambda_0+1)$.
    By the definition of order we have two cases.
    If $f^{i_0}(\lambda_0) > f^{i_0+1}(\lambda_0)$ then the property holds as
    $f^{i_0+1}(\lambda_0)\geq f^{i_0+1}(\lambda_0+1)$.
    The second case is when $f^{i_0}(\lambda_0) = f^{i_0+1}(\lambda_0)$ and
    $S^{i_0}(\lambda_0)>S^{i_0+1}(\lambda_0)$.
    But then $S^{i_0+1}(\lambda_0) \neq \top$ and therefore
    $f^{i_0+1}(\lambda_0) > f^{i_0+1}(\lambda_0+1)$ by the
    condition~\eqref{eq:vertex-condition} on vertices of $U^\b$.
    So the property holds in this case too. 

    Now, repeating the same argument on the suffix $(f^{i},S^{i})_{i > i_0}$ we
    find $\lambda_1>\lambda_0$ and $i_1>i_0$ such that $f^{i_0}(\lambda_0) >
    f^{i_0+1}(\lambda_0+1) \geq f^{i_1}(\lambda_1)> f^{i_1+1}(\lambda_1+1)$.
    Iterating this construction, we obtain an infinite decreasing sequence of
    ordinals: a contradiction.
    
    \item Let $(f,S),(f',S'),(f'',S'')$ be vertices of $U^\beta$ and let
    $c_\lambda \in C_\lambda$. 
    We consider only the right monotonicity, left monotonicity being similar.
    Assume $(f,S) \re{c_\lambda} (f',S') > (f'',S'')$.
    Using~\eqref{eq:edges-keep-order}, we have the following chain of non-strict inequalities
    \[
            (f,S)_{\leb \lambda} \geq (f',S')_{\leb \lambda} \geq (f'',S'')_{\leb \lambda}
    \]
    and conclude that $(f,S) \re{c_\lambda} (f'',S'')$ if any of them is strict.
    Otherwise, the above are equalities, so the definition of transitions and
    order gives us $S(\lambda) \re {c_\lambda} S'(\lambda) \geq
    S''(\lambda)$ in $U_\lambda^\top$.
    By monotonicity of $U_\lambda^{\top}$ we have $S(\lambda) \re{c_\lambda}
    S''(\lambda)$ giving us the desired $(f,S) \re{c_\lambda} (f'',S'')$.
    
    \item Consider an infinite path $(f^0,S^0) \re{c^0_{\lambda^0}} (f^1,S^1)
    \re{c^1_{\lambda^1}} \dots $ in $U^\beta$ where for all $i$,
    $c^i_{\lambda^i} \in C_{\lambda^i}$.
    Let $w=c^0_{\lambda^0} c^1_{\lambda^1}\dots$.
    The aim is to prove that $w \in W$.
    Let $\lambda_0$ be minimal among the $\lambda^i$'s and distinguish two cases.
    \begin{itemize}
        \item If $\lambda_0$ appears infinitely often among the $\lambda^i$'s, then $w$ is $\lambda_0$-supported, so we must prove that $\pi_{\lambda_0}(w) \in W_{\lambda_0}$.
        Since all $\lambda^i$'s are $\geq \lambda_0$,
        property~\eqref{eq:edges-keep-order} gives $(f^i,S^i)_{\leb \lambda_0}
        \geq (f^{i+1},S^{i+1})_{\leb \lambda_0}$ for all $i$.
        Therefore, thanks to well\=/foundedness, $(f^i,S^i)_{\leb \lambda_0}$
        is eventually constant, say starting from index $i_0$.
        Consider any $i\geq i_0$.
        If $\lambda^i=\lambda_0$ we must have $S^i(\lambda_0)\re{c^i_{\lambda_0}}S^{i+1}(\lambda_0)$. 
        Otherwise, $\lambda^i > \lambda_0$, so we have both $(f^i,S^i)_{\leb \lambda^i} \geq (f^{i+1},S^{i+1})_{\leb \lambda^i}$ and $(f^i,S^i)_{\leb \lambda_0} = (f^{i+1},S^{i+1})_{\leb \lambda_0}$, which implies that $S^{i}(\lambda_0) \geq_{\lambda_0} S^{i+1}(\lambda_0)$.
        Therefore for $i \geq i_0$, we have $S^{i}(\lambda_0) \re{c^i_{\lambda^i}} S^{i+1}(\lambda_0) \in E(U_{\lambda_0}^\top)$ if $\lambda^{i}=\lambda_0$, and $S^{i}(\lambda_0) \geq S^{i+1}(\lambda_0)$ otherwise.
        Thanks to monotonicity of transitions in $U^\top_{\lambda_0}$ we conclude
        that $\pi_{\lambda_0}(w_{\geq i_0})$ labels a path of
        $U_{\lambda_0}^\top$. 
        By universality of $U_{\lambda_0}^\top$, this path satisfies
        $W_{\lambda_0}$.
        By prefix-independence $\pi_{\lambda_0}(w)$ also satisfies $W_{\lambda_0}$.
        
        \item Assume now that $\lambda_0$ appears only
        finitely often among the $\lambda^i$'s.
        Let $i_0$ be the maximal $i$ such that $\lambda^{i}=\lambda_0$. 
        We show that 
        \begin{align*}
          \text{either}\quad& (f^0,S^0)_{\leb\lambda_0}> (f^{i_0+1},S^{i_0+1})_{\leb\lambda_0} \quad \text{or}\\
          & (f^0,S^0)_{\leb\lambda_0}=(f^{i_0+1},S^{i_0+1})_{\leb\lambda_0} \text{ and } f^{i_0}(\lambda_0)>f^{i_0+1}(\lambda_0+1)
        \end{align*}
        Thanks to~\eqref{eq:edges-keep-order}, for all $i$ we have $(f^i,S^i)_{\leb \lambda_0} \geq
        (f^{i+1},S^{i+1})_{\leb \lambda_0}$. Consider two cases.
        If $(f^0,S^0)_{\leb\lambda_0} > (f^{i_0+1},S^{i_0+1})_{\leb
        \lambda_0}$ we are done.
        Otherwise, we must have $(f^{i_0},S^{i_0})_{\leb \lambda_0} =(f^{i_0+1},S^{i_0+1})_{\leb
        \lambda_0}$ and $S^{i_0}(\lambda_0) \re{c^{i_0}_{\lambda^{i_0}}} S^{i_0+1}(\lambda_0) \in
        E(U_{\lambda_0}^\top)$.
        Thus $S^{i_0+1}(\lambda_0) \neq \top_{\lambda_0}$; hence, $f^{i_0}(\lambda_0) >
        f^{i_0+1}(\lambda_0+1)$ by the condition~\eqref{eq:vertex-condition} on vertices.

        In the next step we let $\lambda_1>\lambda_0$ be the minimum $\lambda^i$
        for $i>i_0$, and if it appears only finitely often, define $i_1$ to be
        maximal such that $\lambda^{i_1}=\lambda_1$. 
        Just like above, we obtain:
        \begin{align*}
            \text{either}\quad& (f^{i_0+1},S^{i_0+1})_{\leb\lambda_1}> (f^{i_1+1},S^{i_1+1})_{\leb\lambda_1} \quad \text{or}\\
            & (f^{i_0+1},S^{i_0+1})_{\leb\lambda_1}=(f^{i_1+1},S^{i_1+1})_{\leb\lambda_1} \text{ and } f^{i_1}(\lambda_1)>f^{i_1+1}(\lambda_1+1)
        \end{align*}
        Observe that if the second case occurs, and we have
        $f^{i_0}(\lambda_0)>f^{i_0+1}(\lambda_0+1)$ then we can combine these
        inequalities to $f^{i_0}(\lambda_0)>f^{i_0+1}(\lambda_0+1)\geq
        f^{i_0+1}(\lambda_1)=f^{i_1}(\lambda_1)$ (here the second
        inequality is monotonicity of $f$ as $\lambda_1\geq \lambda_0+1$).

        To finish we observe that this process cannot continue forever. 
        Indeed, the first case cannot occur infinitely often due to
        well-foundedness proved in the first item of the lemma.
        If eventually only the second case occurs then this also leads to a 
        contradiction as we can combine the inequalities we have observed
        above to obtain an infinite strictly decreasing chain
        $f^{i_0}(\lambda_0)>f^{i_1}(\lambda_1)>f^{i_2}(\lambda_2)>\dots$.
        \qedhere
    \end{itemize}
\end{enumerate}
\end{proof}

\paragraph{Compositionality properties.}
For every $\lambda<\a$ we can define the graph $U^\b_{<\lambda}$ in the same way as
$U^\b$, but considering the sequence up to $\lambda$ instead of up to $\a$.
We can also define the graph $U^\b_{\geq\lambda}$ by considering the subsequence
starting from $\lambda$.
In this second case, it will be convenient to assume the vertices of
$U^\b_{\geq\lambda}$ are of the form 
$f:[\lambda,\alpha) \to \beta$ and $S :[\lambda,\alpha) \to \bigcup_{\lambda \leq \lambda'< \alpha} V(U_{\lambda'}^\top)$.

The lemma below proves a useful compositionality property; in some sense it states that our construction extends finite lexicographic products.

\begin{lem}\label{lem:composition}
For all ordinals $\beta,\beta'$ and for $\lambda<\alpha$ it holds that $U_{< \lambda}^\beta \ltimes U_{\geq \lambda}^{\beta'} \to U^{\beta' + \beta}$.
\end{lem}

\begin{proof}
    For $v=((f,S),(f',S')) \in V(U_{<\lambda}^\beta \ltimes
    U_{\geq \lambda}^{\beta'})$, we define $\phi(v)=(g,R)$ by
    \begin{equation*}
        g(\lambda')=\begin{cases}
            \beta' + f(\lambda')& \tif \lambda' < \lambda\\
            f'(\lambda')& \tow
        \end{cases}\qquad \tand  \qquad 
        R(\lambda')=\begin{cases}
            S(\lambda') &\tif \lambda' < \lambda\\
            S'(\lambda') &\tow.
        \end{cases}
    \end{equation*}
    It is direct to check that $(g,R) \in V(U^{\beta' + \beta})$; in particular
    $g$ is non-increasing since both $f$ and $f'$ are and values of $f'$
    are $< \beta'$.
    To show that $\phi$ defines a morphism from $U_{<\lambda}^\beta \ltimes
    U_{\geq \lambda}^{\beta'}$ to $U^{\beta' + \beta}$, we pick an edge
    $v_0 = ((f_0,S_0),(f'_0,S'_0)) \re{c_{\lambda'}} ((f_1,S_1),(f'_1,S'_1)) = v_1$ with
    $c_{\lambda'} \in C_{\lambda'}$.
    We write $\phi(v_0) = (g_0,R_0)$ and $\phi(v_1) = (g_1,R_1)$.
    By the definition of $\ltimes$ this edge comes from one of the three cases. 

    \begin{itemize}
    \item If $\lambda'<\lambda$, then $(f_0,S_0) \re{c_{\lambda'}} (f_1,S_1) \in E(U_{< \lambda}^\beta)$, that is,
    \[
        (f_0,S_0)_{\leb \lambda'} > (f_1,S_1)_{\leb \lambda'} \tor [(f_0,S_0)_{\leb \lambda'} = (f_1,S_1)_{\leb \lambda'} \tand S_0(\lambda') \re{c_{\lambda'}} S_1(\lambda') \in E(U_{\lambda'}^\top)].
    \]
    We have $(g_0,R_0)_{\leb \lambda'} = (\beta' + f_0,S_0)_{\leb \lambda'}$, and likewise $(g_1,R_1)_{\leb \lambda'} = (\beta' + f_1,S_1)_{\leb \lambda'}$, so the result follows.
    \item If $\lambda' \geq \lambda$ and $(f_0,S_0) > (f_1,S_1)$, then we have $(g_0,R_0)_{< \lambda} > (g_1,R_1)_{< \lambda}$, which implies $(g_0,R_0)_{\leb \lambda'} > (g_1,R_1)_{\leb \lambda'}$; thus, $(g_0,R_0) \re{c_{\lambda'}} (g_1,R_1)$.
    \item Otherwise, $\lambda' \geq \lambda$, $(f_0,S_0) = (f_1,S_1)$ and $(f_0',S_0') \re{c_{\lambda'}} (f_1',S_1') \in E(U_{\geq \lambda}^{\beta'})$, which rewrites as
    \[
        (f_0',S_0')_{\leb \lambda'} > (f_1',S_1')_{\leb \lambda'} \tor [(f_0',S_0')_{\leb \lambda'} = (f_1',S_1')_{\leb \lambda'} \tand S_0'(\lambda') \re{c_{\lambda'}} S_1'(\lambda') \in E(U_{\lambda'}^\top)].
    \]
    (In the line above, the notation $(f'_0,S'_0)_{\leb \lambda'}$ refers to maps of the form
    $[\lambda,\a) \to \beta'$ and $[\lambda,\alpha) \to \bigcup_{\lambda \leq \lambda''< \alpha} V(U_{\lambda''}^\top)$.)
    Then since $(g_0,R_0)_{< \lambda}=(\beta' + f_0,S_0) = (\beta'+f_1,S_1) =(g_1,R_1)_{<\lambda}$, it follows that
    \[
        (g_0,R_0)_{\leb \lambda'} > (g_1,R_1)_{\leb \lambda'} \tor [(g_0,R_0)_{\leb \lambda'} = (g_1,R_1)_{\leb \lambda'} \tand R_0(\lambda') \re{c_{\lambda'}} R_1(\lambda') \in E(U_{\lambda'}^\top)],
    \]
    the wanted result. \qedhere
    \end{itemize}
\end{proof}

\subsubsection{Universality}

We are now ready to prove our main result.

\begin{thm}\label{thm:universality_min_lexico}
Suppose $(C_\lambda)_{\lambda<\a}$ is a sequence of pairwise disjoint sets of colours, and
$(W_\lambda)_{\lambda<\a}$ is a sequence of prefix-independent objectives with $W_\lambda\incl
C^\w_\lambda$ for all $\lambda$.
Let $\k$ be some cardinal and assume that for every $\lambda<\a$ there is a well-ordered monotone $\k$-universal graph
$(U_\lambda,\geq_\lambda)$ for $W_\lambda$.
Then there is $\b$ such that the power graph $(U^\b,\geq)$ is
$\k$-universal for the min-lexicographic product of $(W_\lambda)_{\lambda<\a}$.
\end{thm}

\noindent
We say that a $C$-graph $G$ \emph{can be mapped} if for some ordinal $\beta$ it
holds that $G \to U^\beta$; otherwise we say that $G$ cannot be mapped. 
Since $U^\beta$ satisfies $W$ (Lemma~\ref{lem:prop_U}), any graph that can be mapped satisfies $W$.
Our goal is to prove the converse: graphs of size $<\k$ that satisfy $W$ can be
mapped.
This implies Theorem~\ref{thm:universality_min_lexico}, by taking $\beta$ large enough so that any graph
smaller than $\kappa$ satisfying $W$ can be mapped into $U^\beta$. 

Our first step is to show that if every graph in a sequence can be mapped then the
directed sum of the sequence can be mapped.

\begin{lem}\label{lem:map_sums}
Let $(G_\mu)_{\mu < M}$ be a family of graphs such that for all $\mu < M$, $G_\mu$ can be mapped.
Then $\ASum_{\mu<M} G_\m$ can be mapped.
\end{lem}

\begin{proof}
    For each $\mu < M$, let $\beta_\mu$ be such that $G_\mu \to U^{\beta_\mu}$.
Let $\f_\m(v)=(f^v_\m,S^v_\m)$ be a morphism $\f_\m:G_\mu \to
U^{\beta_\mu}$.
Recall that $f^v_\m:\a\to\beta_\mu$ and $S^v_\m:\a\to\bigcup_{\lambda<\a}V(U^\top_\lambda)$.
Let $G=\ASum _{\mu<M} G_\m$ and let $\beta = \sum_{\m< M}\beta_\m$.
We define a map $\f: V(G) \to V(U^{\beta})$ by
\begin{equation*}
    \f(v)=\big(\sum_{\m'<\m}\b_{\m'} + f^v_\m,S^v_\m\big),\qquad \text{if $v\in V(G_\m)$ and $\f_\m(v)=(f^v_\m,S^v_\m)$}.
\end{equation*}
It is direct to check that $\f(v)$ is an~element of $V(U^\b)$. 
In particular for the condition~\eqref{eq:vertex-condition} we check that if
$S^v_\m(\lambda)\not=\top$ then
$(\sum_{\m'<\m}\b_{\m'} + f^v_\m)(\lambda)>(\sum_{\m'<\m}\b_{\m'} + f^v_\m)(\lambda+1)$.
This follows directly from the fact that $f^v_\m$ satisfies this condition.

To show that $\f$ defines a morphism $G \to U^\beta$, we take an edge $v \re {c_\lambda} v' \in E(G)$ with $c_\lambda \in C_\lambda$; by definition of $G=\ASum _{\mu<M} G_\m$ there are two cases.
\begin{itemize}
\item The first case is when $v \re {c_\lambda} v' \in E(G_\mu)$ for some $\mu < M$.
Since $\f_\mu$ is a morphism $G_\mu \to U^{\beta_\mu}$, we have $\f_\mu(v)
\re {c_\lambda} \f_\mu(v')$, which rewrites as
\[
    \begin{aligned}
    & (f_\mu^v,S_\mu^v)_{\leb \lambda} > (f_\mu^{v'},S_\mu^{v'})_{\leb \lambda},\ \tor\\ 
    & (f_\mu^v,S_\mu^v)_{\leb \lambda} = (f_\mu^{v'},S_\mu^{v'})_{\leb \lambda},\ %
    \tand\ S_\mu^v(\lambda) \re {c_\lambda} S_\mu^{v'}(\lambda) \in E(U_\lambda^\top).
    \end{aligned}
\]
Since $f^v$ and $f^{v'}$ are obtained by shifting respectively $f_\mu^{v}$ and
$f_\mu^{v'}$ by $\sum_{\mu'<\mu}\beta_{\mu'}$, and $S^{v}=S_\mu^{v}$ and
$S^{v'}=S_\mu^{v'}$, we get that $(f^v,S^v) \re {c_\lambda} (f^{v'},S^{v'})$, as
required.

\item Otherwise, $v \in V(G_{\mu})$ and $v' \in V(G_{\mu'})$ for $\mu>\mu'$. Then
\[
    f^v(0)=\sum_{\mu''<\mu} \beta_{\mu''} + f_\mu^v(0) \ >\ \sum_{\mu'' < \mu'} \beta_{\mu''} + f_{\mu'}^{v'}(0)  = f^{v'}(0),
\]
where the inequality holds since $\beta_{\mu'} > f_{\mu'}^{v'}(0)$.
Therefore, $(f^v,S^v) \re {c_\lambda} (f^{v'},S^{v'})$. \qedhere
\end{itemize}
\end{proof}

\noindent
We now prove that any $C_{\geq \lambda}$-graph that can be mapped into $U^\beta$ can
be also mapped into $U_{\geq \lambda}^{\beta'}$, for some bigger $\beta'$.

\begin{lem}\label{lem:reduction_to_large_lambda}
Let $\lambda<\alpha$ and let $G$ be a $C_{\geq \lambda}$-graph which can be mapped.
Then there is an ordinal $\beta'$ such that $G \to U_{\geq \lambda}^{\beta'}$.
\end{lem}

\begin{proof}
Let $\phi: G \to U^\beta$.
For each $(f',S') \in V(U_{<\lambda}^\beta)$, we let $G_{(f',S')}$ be the restriction of $G$ to vertices in
\[
    \phi^{-1}\{(f,S)\in V(U^\beta) \mid (f,S)_{< \lambda} =(f',S')\}.
\]
(It may be that some of the $(G_{(f',S')})$ are empty; this is not an issue in the proof below.)
We now make two claims which are proved below.

\begin{claim}\label{cl:6}
For each $(f',S')$, it holds that $G_{(f',S')} \to U_{\geq \lambda}^{\beta}$.
\end{claim}

\begin{claim}\label{cl:7}
It holds that $G \to \overset{\leftarrow}{\sum}_{(f',S')} G_{(f',S')}$, where the $(f',S')$'s are ordered by a well-order extending the one of $V(U_{<\lambda}^\beta)$.
\end{claim}

\noindent
Putting the two claims together with Lemma~\ref{lem:map_sums} yields the desired
result.
For Claim~\ref{cl:6}, it suffices to consider the restriction of $\phi$ to $G_{(f',S')}$ (restrictions of morphisms are morphisms).
For Claim~\ref{cl:7}, it suffices to recall
property~\eqref{eq:edges-keep-order} saying  that for every edge $(f_0,S_0) \re 
{c} (f_1,S_1) \in E(U^\beta)$ such that $c \in C_{\geq \lambda}$, we have
$(f_0,S_0)_{< \lambda} \geq (f_1,S_1)_{<\lambda}$.
\end{proof}

\noindent
We now prove a crucial ingredient to the proof of Theorem~\ref{thm:universality_min_lexico}.
Recall that $G[v]$ is the restriction of $G$ to vertices reachable from $v$.

\begin{lem}\label{lem:sublem1}
    If $G$ satisfies $W$ and cannot be mapped then there is $v \in V(G)$ such that $G[v]$ cannot be mapped.
    Moreover, $v$ can be picked so that it has a predecessor in $G$.
\end{lem}

\begin{proof}[Proof of Lemma~\ref{lem:sublem1}]
Assume for contradiction that for all $v \in V(G)$, $G[v]$ can be mapped.
Take a well-ordering $(v_{\mu})_{\mu <M}$ of all vertices of $G$, and define $G_\mu$ to be
\[
    G_\mu = G[v_\mu] - \bigcup_{\mu'<\mu} V(G[v_{\mu'}]).
\]
Then for each $\mu <M$, it holds that $G_\mu \to G[v_\mu]$; therefore, $G_\mu$ can be mapped.
Hence, by Lemma~\ref{lem:map_sums}, $\ASum_{\m<M}G_\m$ can be mapped.
But $G \to \ASum_{\m<M}G_\m$, so $G$ can also be mapped: a contradiction proving
that there is $v \in V(G)$ such that $G[v]$ cannot be mapped. 

We now show that $v$ can be taken to have a predecessor, so assume that the $v$
constructed above does not have a predecessor (in particular, there is no loop
around $v$). 
We show that for some successor $v'$ of $v$, $G[v']$ cannot be mapped.
Observe that if $v$ is reachable from some of its successors $v'$ then
$G[v]=G[v']$.
So we can take $v'$ in this case.
Otherwise, we adapt the argument from the previous paragraph. 
Assume that $G[v']$ can be mapped for all successors $v'$ of $v$.
Well-order them as $(v'_{\mu})_{\mu<M'}$, and let
\[
    G'_\mu = G[v'_\mu] - \bigcup_{\mu'<\mu} V(G[v'_{\mu'}]).
\]
Then the $G'_\mu$'s can be mapped.
Let $G_{M'}$ be the restriction of $G$ to $\{v\}$.
Since $v$ is not reachable from any of its successors, there is no loop around
$v$. 
So $G_{M'}$ is edgeless, and therefore it can be mapped.
Now observe that
\[
    G[v] \to \ASum_{\m<M'}G'_\m \ap G_{M'},
\]
which can be mapped thanks to Lemma~\ref{lem:map_sums}. 
This contradiction shows that this case is impossible.
\end{proof}

\noindent
We are now ready to present an inductive proof of Theorem~\ref{thm:universality_min_lexico}.

\begin{proof}[Proof of Theorem~\ref{thm:universality_min_lexico}]
The proof is by induction on the ordinal $\alpha$; therefore, we assume the result known for ordinals $< \alpha$:
\begin{center}
For any $\lambda < \alpha$ and for any $C_{< \lambda}$-graph $G$ of size $<\k$ satisfying $W$, there is $\beta$ such that $G \to U_{< \lambda}^\beta$.
\end{center}
Let $G$ be a $C$-graph of size $<\k$ satisfying $W$ and assume for contradiction that $G$ cannot be mapped.

\begin{claim}\label{cl:10}
For any $\lambda<\alpha$, the restriction $G^{\geq \lambda}$ of $G$ to edges with colours in $C_{\geq \lambda}$ cannot be mapped.
\end{claim}

 \begin{proof}
  Suppose for a contradiction that there is $\lambda$ such that $G^{\geq \lambda}$ can be
  mapped.
  We construct a $C_{< \lambda}$-graph $G'$ that can be mapped, say into $U_{<\lambda}^{\beta'}$, and use $U_{<\lambda}^{\beta'} \ltimes G^{\geq \lambda}$ to show that $G$
  can be mapped. 

  The $C_{< \lambda}$-graph $G'$ has the same vertices as $G$,  $V(G')=V(G)$.
  The edges have colours  $c \in C_{<\lambda}$, and are given by:
  \[
      v \re {c} v' \in E(G')\quad\text{when}\quad \exists u,u' \in V(G), \quad v \rp{C_{\geq\lambda}^*} u \re{c} u' \rp{C_{\geq \lambda}^*} v' \tin G, 
  \]
  where the notation $v \rp{C^*_{\geq \lambda}} u$ means that there is a path in
  $G$ from $v$ to $u$ using only edges with colours in $C_{\geq \lambda}$
  (stated differently, a path in $G^{\geq \lambda}$).
  Note that these paths may be empty, and in particular, edges of $G$ with
  colour in $C_{<\lambda}$ also belong to $G'$. 

Let us prove that $G'$ satisfies $W$.
Consider an infinite path $\pi' = v_0 \re{c_0} v_1 \re{c_1} \dots$ in $G'$; it is labelled by the word $w'=c_0 c_1 \dots \in C_{<\lambda}^\omega$.
Then there is a path of the form
\[
    \pi = v_0 \rp{} u_0 \re{c_0} u'_0 \rp{} v_1 \rp{} u_1 \re{c_1} u'_1 \rp{} \dots,
\]
in $G$, where paths $v_i \rp{} u_i$ and $u'_i \rp{} v_{i+1}$ are labelled by colours in $C_{\geq \lambda}$.
Since $G$ satisfies $W$, the label $w$ of $\pi$ belongs to $W$, in particular it
is $\lambda'$-supported for some $\lambda'$, and since $w$ has
infinitely many occurrences of letters from $C_{<\lambda}$, it must be that
$\lambda'<\lambda$.
Thus $w'$ is also $\lambda'$-supported and $\pi_{\lambda'}(w') = \pi_{\lambda'}(w) \in W_{\lambda'}$ and thus $w' \in W$.
Therefore, $G'$ satisfies $W$; hence, we obtain by induction that $G' \to U_{<\lambda}^{\beta'}$ for some $\beta'$.

Since $G'\to U_{<\lambda}^{\beta'}$ we can find a minimal morphism 
$\phi'\colon G' \to U_{< \lambda}^{\beta'}$.
This means that it is a morphism not pointwise bigger than any other morphism $G'\to
U_{<\lambda}^{\beta'}$.
Such a morphism has a property that for any pair $(v,v')$ of vertices, if for every colour $c$, all $c$-successors of $v'$ are also $c$-successors of $v$, then $\phi'(v)\geq\phi'(v')$ (otherwise we could obtain a smaller morphism by mapping $v'$ to $\phi'(v)$).

We now show that $G \to U_{<\lambda}^{\beta'} \ltimes G^{\geq \lambda}$.
Consider the map $\phi$ between these graphs given by $\phi(v) = (\phi'(v),v)$, which we show to be a morphism.
Take an edge $v \re c v' \in E(G)$.
\begin{itemize}
\item If $c \in C_{< \lambda}$, then $v \re c v' \in E(G')$; thus, $\phi'(v) \re c \phi'(v') \in E(U_{<\lambda}^{\beta'})$, which implies the result.
\item Otherwise, $c \in C_{\geq \lambda}$. 
Then in $G'$, $\out(v)\supseteq\out(v')$.
By the above-mentioned property of minimal morphisms, this implies that $\phi'(v) \geq
\phi'(v')$. 
Together with the fact that $v \re c v' \in E(G^{\geq \lambda})$, this implies that $(\phi'(v),v) \re c (\phi'(v'),v') \in E(U_{<\lambda}^{\beta'} \ltimes G^{\geq \lambda})$, as required.
\end{itemize}
Thus $G \to U_{<\lambda}^{\beta'} \ltimes G^{\geq \lambda}$.
Now if $G^{\geq \lambda}$ could be mapped, then by Lemma~\ref{lem:reduction_to_large_lambda} we get $G^{\geq \lambda} \to U_{\geq \lambda}^\beta$; therefore, it follows from Lemma~\ref{lem:composition} that $G$ can be mapped, a contradiction.
 \end{proof}

If $\alpha=\gamma+1$ is a successor ordinal, Claim~\ref{cl:10} for $\lambda=\gamma$ already gives a contradiction: $G^{\geq \gamma}$ is a $C_\gamma$-graph satisfying $W$, hence $W_\gamma$, so it maps into $U_\gamma$, and thus into $U^2$ via $u \mapsto (f,S)$, where $f$ is constant equal to $1$, $S(\gamma)=u$, and $S(\lambda)=\top_\lambda$ for $\lambda<\gamma$.
Hence, we assume from now on that $\alpha$ is a limit ordinal.

Let $G_0=G$ and let $v_0$ be such that $G_0[v_0]$ cannot be mapped, obtained from Lemma~\ref{lem:sublem1} (here, the fact that $v_0$ has a predecessor in $G_0$ is not used).
We will construct a decreasing sequence of subgraphs $G_0,G_1,\dots$ of $G$ and
vertices $v_0,v_1, \dots$ with non-empty paths 
$\pi_i$ from $v_i$ to $v_{i+1}$ in $G_i$, with the property that for all $i$, all edges in $G_{i+1}$ (and therefore also in subsequent graphs) have colours in $C_{>\lambda_i}$, where $\lambda_i$ is the maximal colour of an edge in $\pi_{<i} = \pi_0 \dots \pi_{i-1}$.
This implies the desired contradiction as the label of $\pi$ is not supported, and thus does not satisfy $W$.
The crucial invariant in the construction is that the $G_i[v_i]$'s cannot be mapped.

\begin{figure}[ht]
\begin{center}
\includegraphics[width=0.55\linewidth]{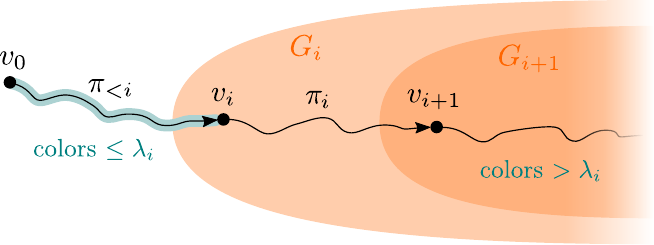}
\end{center}
\caption{Constructing a path violating $W$.}\label{fig:final_construction}
\end{figure}

Assume that the path up to $v_i$ has been constructed (see also Figure~\ref{fig:final_construction}), and let $\lambda_{i}$ be as above (or $\lambda_i=0$ if $i=0$).
Since $G_i[v_i]$ cannot be mapped, Claim~\ref{cl:10} says that $G_i[v_i]^{\geq \lambda_i +1}$ cannot be mapped.
So we let $G_{i+1}=G_i[v_i]^{\geq \lambda_i +1}$, and then apply Lemma~\ref{lem:sublem1} to $G_{i+1}$ to obtain $v_{i+1} \in V(G_{i+1})$ such that $G_{i+1}[v_{i+1}]$ cannot be mapped and $v_{i+1}$ has a predecessor $u_{i+1}$ in $G_{i+1}$.
Since $G_{i+1}$ is a subgraph of $G_i[v_i]$, there is a path $\pi_i$ in
$G_i$ from $v_i$ to $v_{i+1}$, which we can take to go through $u_{i+1}$. This
ensures the path is non-empty.
\end{proof}

\noindent
Our main result, Theorem~\ref{th:main-pos-min-lex}, follows from Theorem~\ref{thm:universality_min_lexico} and Lemma~\ref{lem:prop_U}.

\section{Conclusions}\label{sec:conclusions}

In this work, we have introduced two positionality-preserving operations on objectives  generalising lexicographic products to families indexed by ordinals: max- and min-lexicographic products. These two operations extend our understanding of positionality in two orthogonal manners.

Max-lexicographic products yield a natural generalisation of the $\Parity$ languages, providing
a family of positional languages that are complete for infinite levels of the difference hierarchy over $\bsigma{2}$ (Theorem~\ref{thm:difference-complete}).
This sets a first stone in the systematic study of positional objectives within $\bdelta{3}$, the natural topological generalisation of $\omega$-regular objectives.

Min-lexicographic products, on the other hand, easily go beyond $\bdelta{3}$.
They provide a tool to show positionality of objectives in $\bsigma{3}$ (such as $\omegaBuchi$), the highest level in the Borel hierarchy in which positional objectives have been found.\footnote{During the preparation of this manuscript, a positional $\bpi{3}$-complete objective has also been proposed~\cite{COV24Pi3}.}
An interesting question is whether there are positional objectives in all the levels of the Borel hierarchy.

Furthermore, we have proved a special case of Kopczyński's conjecture, namely, closure of positionality under colour-increasing unions of objectives (Theorem~\ref{thm:weak_kopczynski}).
The lexicographic product of a family of objectives provides a sort of under-approximation to their union. Whether the positionality of lexicographic products can help to resolve the general case of Kopczyński's conjecture is an exciting open problem.

\bibliographystyle{alphaurl}
\bibliography{bib}

\end{document}